\documentclass[11pt ]{article}
\usepackage{color}
\usepackage{amssymb}

\usepackage{enumerate}
\usepackage{amsmath}
\usepackage{natbib}
\usepackage{graphicx}

\usepackage{subfigure}
\usepackage[margin=1.1in]{geometry}

\usepackage{bbm}
\usepackage{adjustbox}
\usepackage{multirow}

\newtheorem{theorem}{Theorem}

\newtheorem{example}[theorem]{Example}

\newtheorem{proposition}[theorem]{Proposition}
\newtheorem{remark}[theorem]{Remark}

\newenvironment{proof}[1][Proof]{\textbf{#1.} }{\ \rule{0.5em}{0.5em}}

\newenvironment{mat}{\left[\begin{array}{ccccccccccccccc}}{\end{array}\right]}
\newcommand\bcm{\begin{mat}}
\newcommand\ecm{\end{mat}}

\newenvironment{rmat}{\left[\begin{array}{rrrrrrrrrrrrr}}{\end{array}\right]}
\newcommand\brm{\begin{rmat}}
\newcommand\erm{\end{rmat}}

\newcommand{\EE}{{\mathord{I\kern -.33em E}}}
\def\E{{\EE}} % expectation
\def\Q{{\mathbb Q}} % rationals
\def\R{{\mathbb R}} % reals
\def\1{{\mathbf 1}} % indicator
\def\F{{\mathcal F}} % potential measure
\begin{document}
\title{A Top-Down Approach for the Multiple  Exercises and Valuation of   Employee Stock Options}
\author{Tim Leung\thanks{Department of Applied Mathematics, University of Washington, Seattle WA 98195. E-mail:
\mbox{timleung@uw.edu}. Corresponding author.}    \and Yang Zhou\thanks{Department of Applied Mathematics, University of Washington, Seattle WA 98195. E-mail:
\mbox{yzhou7@uw.edu}.}}

\date{\today} \maketitle
\begin{abstract}We propose a new framework to value employee stock options (ESOs) that captures multiple exercises of different quantities over time. We also model the ESO holder's job termination risk and incorporate its impact on the payoffs of both vested and unvested ESOs. Numerical methods based on Fourier transform and finite differences are developed and implemented to solve the associated  systems of PDEs. In addition, we introduce a new valuation method  based on maturity randomization  that yields  analytic formulae for vested and unvested ESO costs. We examine the cost impact of job termination risk, exercise intensity, and various contractual features.
\end{abstract}

\newpage
\section{Introduction}
%An employee stock option that grants specified employees of a company the right to buy a certain amount of company shares at a predetermined price for a specific period. The main difficulty  of valuation lies in the uncertain timing of contract
%termination due to various causes, and a number of contractual features further
%complicate the problem. For instance, the employee's sudden departure from the firm
%will lead to either early exercises or forfeiture of ESOs.

The use of employee stock options (ESOs) as part of compensation is a common practice among large and small companies in the   United States. Financial Accounting Standards Board (FASB) requires companies to value  these stock options and report the total granting cost.\footnote{See FASB Accounting Standards Codification (ASC) no.718 (formerly, FASB Statement 123R), \emph{Accounting for Stock-Based Compensation}.} This requirement raises the need for valuation methods that can effectively capture the payoff structure  and exercise pattern of these  stock options.

Empirical studies suggest that  ESO holders tend to start exercising their options exercise early, often soon after the vesting period, and gradually exercise the remaining options over multiple dates before maturity. \cite{HuddartLang1996}, \cite{Marquardt2002}, and \cite{Bettis2005} point out that, for ESOs with 10 years to maturity, the expected time to exercise is 4 to 5 years. Investigating  how ESO exercises are spread out over time, \cite{HuddartLang1996} show that the mean fraction of options exercised by a typical employee at one time varied from 0.18 to 0.72.   For more empirial studies,  we refer to  \cite{HuddartLang1996}, \cite{BettisBizjakLemmon2001}, \cite{Marquardt2002}, \cite{Armstrong07}, \cite{Kevin08},  \cite{RandallErik} and \cite{Carpenter11}. These empirical findings motivate us to consider a valuation model that account for multiple exercises of various units of options at different times.  As noted by \cite{JainSubramanian2004}, ``the incorporation of multiple-date exercise has important economic and account consequences."

% To maintain the employees in the firm, the firm usually imposes a vesting period that employees could not exercise the option. In the vesting period, the ESO is called \emph{vested}. After vesting period, ESO becomes \emph{unvested} and can be exercised at any time before it getting expired. During the vesting period, if the employee leaves the firm voluntarily or involuntarily, he has to forfeit his options.  

%Due to the wide use of ESOs, the Financial Accounting Standards Board (FASB) has become concerned about the cost of ESOs to shareholders. In 2004, FASB passed \emph{Statement of Financial Accounting Standards No.123 (revised)}, mandating firms to report "the grant-date fair value" of the issued ESOs.

In this paper, we take the firm's perspective to determine the cost of an ESO grant. An ESO grant commonly  involves multiple options with a long maturity. There is also a vesting period, during which option exercise is prohibited and job termination leads to forfeiture of the options. The key component of our proposed valuation framework is an exogenous jump process that models the random exercises over time. Within our framework, the employee's exercise intensity can be   constant or stochastic, and the number of options exercised at each time can be specified to be deterministic or random. In essence, this top-down approach offers a flexible setup to model any exercise pattern.  The idea is akin to the top-down approach in credit risk (\cite{GieseckeGoldberg}), where the exogenous jump process represents portfolio losses.   Since the  ESO payoff depends heavily on when the employee leaves the firm, we  also include a random job termination time and allow the job termination rate to be different during and after vesting period. 

The valuation problem leads to the study of the system of partial differential equations (PDEs) associated with the vested and unvested options. In order to compute the ESO costs, we present two numerical methods to solve the PDEs. We   discuss the method of fast Fourier transform (FFT),  followed by the finite difference method (FDM).  By applying Fourier transform, we simplify  the original second-order PDEs to ODEs in the constant intensity case and first-order PDEs in the stochastic intensity case. The ESO costs are recovered via inverse  fast Fourier   transform. The results from the two methods are illustrated and compared under both deterministic and stochastic exercise intensities. Furthermore,  we introduce a new valuation method  based on maturity randomization. The key advantage of this method is that it yields  analytic formulae, allowing  for instant computation.  

Using all three numerical methods, we compute the costs and examine the impact of job termination risk, exercise intensity, vesting period, and other features. Among our findings, we illustrate the distributions of exercise times under different model specifications, and also  show that the average time of exercises tends to increase nonlinearly with the number of ESOs granted, resulting in a higher per-unit cost.  In other words, under the assumption that the ESOs will be exercised gradually, a larger ESO grant has an indirect effect of delaying exercises, and thus leading to higher ESO costs.

% \cite{JennergrenNaslund1993}, \cite{HemmerMatsunagaShevlin1994}

%In Section 2, we formulate the ESO valuation under time-dependent exercise intensity framework. In Section 3, we develop and implement a finite difference method and fast Fourier transform method. In Section 4, we discuss the effects of exercise intensity and multiple options in evaluating the ESO cost to the firm. In Section 5, we introduce maturity randomization method. In Section 6, we compute the ESO cost under stochastic exercise intensity framework via FDM and FFT. Section 6 concludes this paper.
%\begin{itemize}\item Employees tend to gradually exercise fractions of ESOs prior to maturity (which is pointed out in utility-based models as well, but in intensity models exercise times are unpredictable),
%\item Exercise fraction is dependent of stock price level and its movement,
%\item Options are exercised earlier with higher dividend and stock price volatility.
%\end{itemize}
%\tim{Write a parag. about exactly what we do in this paper. Our approach and all model features.  What we compute. PIDEs. How we solve the PIDEs to compute costs. Stochastic intensity. }

 In the literature, there are three main approaches for risk-neutral valuation and expensing of ESOs. Models are often be differentiated by their assumptions on exercise timing. One approach is to pre-specify an exercise boundary that determines the employee's exercise strategy  of a single ESO. In turn, the ESO is then priced as an option of barrier type (\cite{HullWhite2004,Cvitanic2008}). The boundary is typically chosen to be explicit and simple for the ease of computation but does not come with empirical or behavioral justification.  
 
 Another approach is to an optimal exercise time that maximizes the expected discounted payoff under some risk-neutral pricing measure (\cite{LeungWan}). Instead of the risk neutrality assumption, a number of related studies incorporate the employee's risk preferences and hedging constraints and derive the optimal exercise strategy by solving a utility maximization problem. For this line of research, we refer to \cite{JainSubramanian2004,GrasselliHenderson2009,Leung_MF09,Leung_SIAM09}, and \cite{Carmona11}.  In particular, \cite{JainSubramanian2004} and \cite{Leung_MF09}, respectively, propose discrete-time and continuous-time models that allow the risk-averse employee to strategically exercise the ESOs over time rather than all on the same date. In addition to accounting for multiple-date exercises, \cite{GrasselliHenderson2009} also show that a risk-averse employee may find it optimal to exercise multiple options simultaneously at different exercise times.
  
% They show that these methods fail to capture the fractional exercise behavior that employees tend to exercise large fraction of ESO at one time. In addition, they show that fractional exercise behavior appears when costly exercise is incorporated.   \cite{JainSubramanian2004}  propose a discrete-time model to value ESOs allowing for the possibility that a risk-averse employee strategically exercises her options over time rather than at a single date. 

In reality, firms do not know when ESOs will be exercised. Therefore, it is reasonable to model ESO exercises as some exogenous events so that the firm is not assumed to have access to the employee's risk preferences and exercise strategy. This leads to the approach, as studied by \cite{JennergrenNaslund1993,CarrLinetsky2000} among others, that models ESO exercise by  the first arrival time of an exogenous jump process. Although the exercises are exogenous events, the frequency and timing of their exercises can be dependent on the firm's stock and other contractural features. Our proposed approach is essentially an extension of this approach to modeling multiple ESO exercises over the life of the options.

%Alternatively,  use the first jump time of some exogenous Poisson process to encapsulate all possible causes of a single early exercise. In addition, \cite{CarrLinetsky2000} also proposes a single-ESO  valuation  model  in which exercise intensity and job termination risk depend on the company stock price and time. These models do not capture multiple exercises and fractional exercise behavior. \cite{Carmona11} introduce job termination rate and compute the associated perpetual ESO cost.   In contrast to our proposed model, these prior studies do not account for multiple options and  fractional exercise behavior.

%%{Model comparison}
%In contrast to these models, we study an ESO valuation framework that captures employees' early fractional exercise behavior. As opposed to the first heuristic method we mentioned above, we assume the firm has no access to employees' optimal exercise strategy. Instead, we construct an intensity-based framework, which is inspired by the second heuristic method. However, both \cite{JennergrenNaslund1993} and \cite{CarrLinetsky2000} only consider the single option. In addition, 

The rest of the paper is organized as follows. In Section \ref{sect-ESOmodel}, we present our ESO valuation model. 
The numerical method is discussed in Section \ref{sect-num}. In Section \ref{sect-stochastic}, we discuss the case with stochastic exercise intensity. Then in Section \ref{sect-mat-rand} we introduce a novel valuation method based on maturity randomization. Finally, concluding remarks are provided in Section \ref{sect-conclude}.

\section{ESO Valuation Model}\label{sect-ESOmodel}
We begin by describing the ESO payoff structure, and then introduce the stochastic model that captures  various sources of randomness. The   valuation of both vested and unvested ESOs is presented.

\newpage
\subsection{Payoff Structure}
 The ESO is  an early exercisable  call option written on the company stock with a long maturity $T$ ranging from 5 to 10 years.  In order to maintain the incentive effect of ESOs, the company typically prohibits the ESO holder (employee) from exercising during a  {vesting period} from the grant date. During the vesting period, which ranges from 1 to 5 years, the holder's departure from the company, voluntarily or forced, will lead to forfeiture of the option, rending it worthless.  We denote [$0,t_v$) as the vesting period, and after  the date $t_v$ the ESO is \emph{vested} and free to be exercised until it expires at time $T$.   The ESO payoff at any time $\tau$ is $(S_{\tau} -K)^+1_{\{t_v\le\tau\le T\}}$, where $S_{\tau}$ is the firm's stock price at time $\tau$ and $K$ is the strike price. Upon departure, the employee is supposed to exercise all the remaining options.   Figure \ref{ESOpayoff} shows all four payoff scenarios associated with an ESO. \\

\begin{figure}[ht]
  \centering
  \includegraphics[width=4.5in]{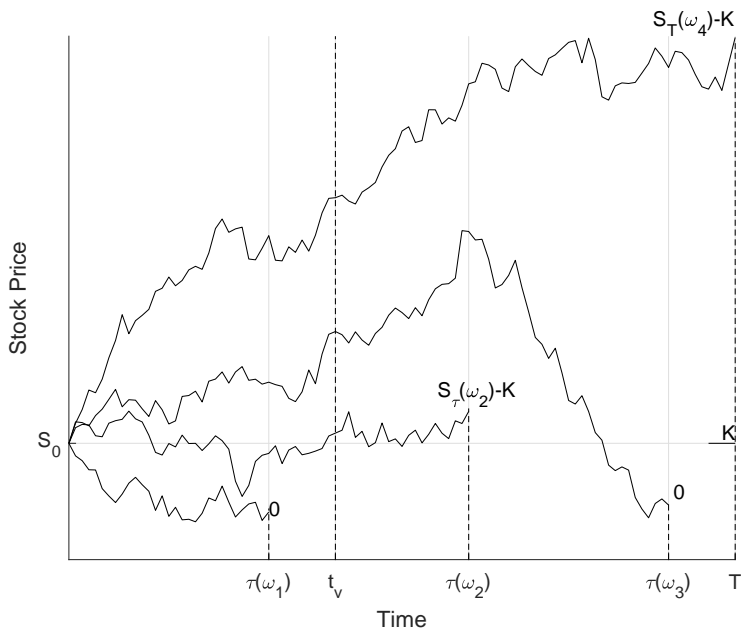}\\
  \caption{ESO payoff structure. From bottom path to top path: (i) The employee leaves the firm during the vesting period, resulting in forfeiture of the ESO and a zero payoff. (ii) The employee exercises the vested ESO before maturity due to desire to liquidate or job termination and receive the payoff $(S_\tau(\omega_2)-K)^+$. (iii) The employee exercises the vested ESO before maturity due to job termination, but receives nothing. (iv) The employee exercises the option at  maturity $T$. }\label{ESOpayoff}
\end{figure}
 
 \vspace{10pt}
%We denote by  $\mathbb{F}\equiv(\mathcal{F}_t)_{t\ge 0}$   the filtration generated by $W^\Q$.

\subsection{Job Termination and Exercise Process}
The employee's job termination plays a crucial role in the exercise timing and resulting payoff of the ESOs. 
We model the job termination time during the vesting period by an exponential random variable $\zeta\sim\exp(\alpha)$, with  $\alpha\ge0$. When the ESO becomes vested after  $t_v$, we   model the employee's job termination time  by another   exponential random variable $\xi\sim\exp(\beta)$, with $\beta\ge0$. We assume that $\zeta$ and $\xi$ are mutually independent. This approach of modeling job termination by an exogenous random variable  is also used by \cite{JennergrenNaslund1993}, \cite{Carpenter1998}, \cite{CarrLinetsky2000}, \cite{HullWhite2004}, \cite{SircarXiong2007}, \cite{Leung_SIAM09}, \cite{Carmona11}, and \cite{LeungWan}, among others. In our model, using two different exponential times allows us to account for the varying level of job termination risk during and after the vesting period. 

An ESO grant typically contains multiple options. Empirical studies show that employee tends to exercise the options gradually over time, rather than exercising all options at once.  This motivates us to model the sequential random timing of exercises. In our proposed model, we consider a grant of $M$ units of identical early exercisable ESOs with the same strike price $K$ and expiration date $T$.  These $M$ ESOs are exercisable only after the vesting period $[0,t_v)$. For the vested ESOs, we define the random \emph{exercise process} $L_t$, for ${t_v\leq t\leq T}$, to be the positive  jump process representing  the number of ESOs exercised over time. As such,  $L_t$ is an integer process that takes value on $[0,M]$. The corresponding jump times are denoted by the sequence $(\tau_1, \tau_2, \ldots)$, and the frequency of exercises is governed by the jump intensity process $(\lambda_t)_{t_v\leq t\leq T}$.  

The jump size for the $i$th jump of $L$ represents the number of ESOs exercised and is described by  a discrete random variable  $\delta_i$. The exercise process starts at time $t_v$ with  $L_{t_v}=0$. By definition, we have $L_T\le M$. This means that the random  jump size at any time $t$  must take value within $[1,M-L_{t-}]$. Also, as soon as $L_t$ reaches the upper bound $M$, the jump intensity  $\lambda_t$ must be set to be zero thereafter. Given that the employee still holds $m$   options,  the probability mass function of the random jump size is  \begin{equation}
p_{m,z}\triangleq\mathbb{P}\{\delta_i=z|L_{\tau_i-}=M-m\}\,.
\end{equation} 
In turn, the expected number of options to be  exercised at each exercise time is given by 
\begin{align}
\bar{p}_m\triangleq\sum_{z=1}^m z p_{m,z}, \label{barpm}
\end{align} 
which again depends on the current number of ESOs held.

 The employee may exercise single or multiple units of ESOs over time.  On the date of expiration or job termination,  any unexercised options must be exercised.  Hence, the discounted  payoff from the ESOs over $[0,T]$ is a sum of two terms, given by
\begin{equation}\label{payoff1}
\bigg(\int_{t_v}^{T\wedge\xi} e^{-r t} (S_t -K)^+ dL_t + e^{-r(T\wedge\xi)} (M- L_{T\wedge\xi}) (S_{T\wedge\xi}-K)^+\bigg)1_{\{\zeta\ge t_v\}}.
\end{equation} 
The indicator $1_{\{\zeta\ge t_v\}}$ means that the ESO payoff is zero if the employee leaves the firm during the vesting period.

 \vspace{10pt}
 
%If we assume that the risk of early exercises and random jump sizes are unpriced, then we can compute the ESO cost by
%taking expectation under the risk-neutral measure. Otherwise, we have to look into how to hedge ESO payoff from the
%firm's perspective in order to determine the cost. 

\begin{example}[Unit Exercises]\label{PoiExer}Suppose $L_t$ be a nonhomogeneous Poisson process $(N_t)_{0\le t\le T}$ with a time-varying  jump intensity function $\lambda(t)$, for $0\le t\le T$. At each jump time a single option is exercised. In Figure \ref{PoissonExercise} we illustrate three possible scenarios. In  scenario (i), the employee   exercises 6 out of 10 options one by one, but must  exercise 4 remaining options upon job termination realized at time $\xi(\omega_1)$. In scenario (ii) the employee  exercises  all 10 options one by one before   maturity. In scenario (iii), the employee has not exercised all the options by maturity, so all remaining options are exercised at time $T$.  \end{example}
\clearpage

\begin{figure}[h]
\begin{centering}
%\begin{subfigure}[b]{0.5\textwidth}
\includegraphics[width=4in]{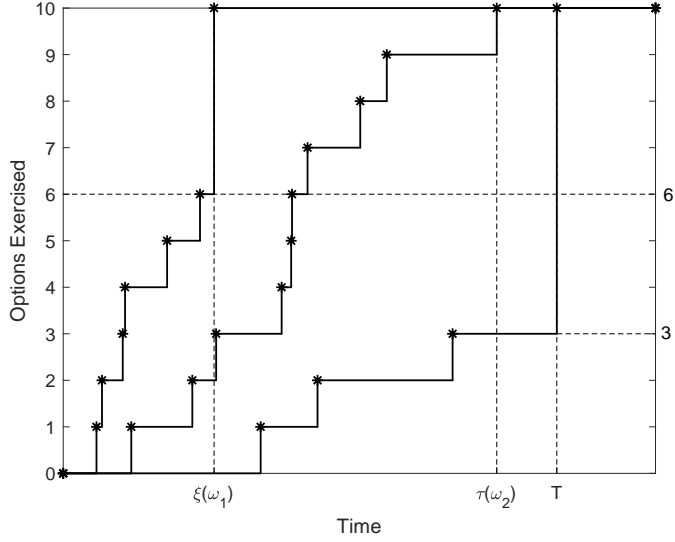}
%\end{subfigure}
%~
%\begin{subfigure}[b]{0.5\textwidth}
%\includegraphics[width=\textwidth]{MultipleExercise.eps}
%\caption{Multiple Exercise}
%\end{subfigure}
\caption{Three illustrative sample paths of the process for Poisson exercises of 10 ESOs.  From path top to bottom path: (i) The employee first exercises 6 out of 10 options one by one, but is then forced to exercise 4 remaining options upon job termination realized at time $\xi(\omega_1)$. (ii) The employee exercises all the options one by one before expiration and job termination. The last option is exercised at $\tau(\omega_2)$ shown in the plot. (iii) The employee exercises 3 options one by one before maturity and 7 remaining options at maturity.}\label{PoissonExercise}
\end{centering}
\end{figure}

\begin{example}[Block Exercises]\label{MultiExer} Suppose the employee can exercise one or more options at each exercise time. As an example, We assume a uniform distribution for the number of options to be exercised, so we set $p_{m,z}={m}^{-1}$ for $z=1, \ldots, m$. In Figure \ref{AverageExercise},  we illustrate the distributions of  the weighted average exercise time $\bar\tau$ defined by  \begin{equation}\label{AverageExerciseTime}
\bar\tau=\frac{\sum_{i=1}^N \delta_{i}*\tau_i}{M},
\end{equation}where $\delta_{i}$ is  the   number of ESOs exercised at the $i$th exercise time $\tau_i$, and $N$ is the number of distinct exercise times  before or at time $T$.  For each simulated path, we take an average of the distinct exercise times    weighted by the number of options exercised at each time.   With common parameters $M=20, t_v=0, T=10$, the histograms of  $\bar\tau$ correspond to different values of $\lambda$ and $\beta$. With  a low job termination rate $\beta$ and low exercise intensity $\lambda$ (panel (a) where $\beta=0$, $\lambda=0.3$), more options tend to be exercised at maturity. Comparing panel (b)  to  panel  (c), and also  panel (b) to  panel (d), we see that a higher job termination rate or higher exercise intensity lowers the average exercise time and reduces instances of exercising at maturity.   Similar patterns can also  be found in  empirical studies \citep[Fig. 3]{RandallErik}.

\end{example}

\begin{figure}[h]
\centering
\subfigure[ ]
{\includegraphics[trim={1cm 0cm 1cm 0.1cm},clip,width=2.8in, angle=0]{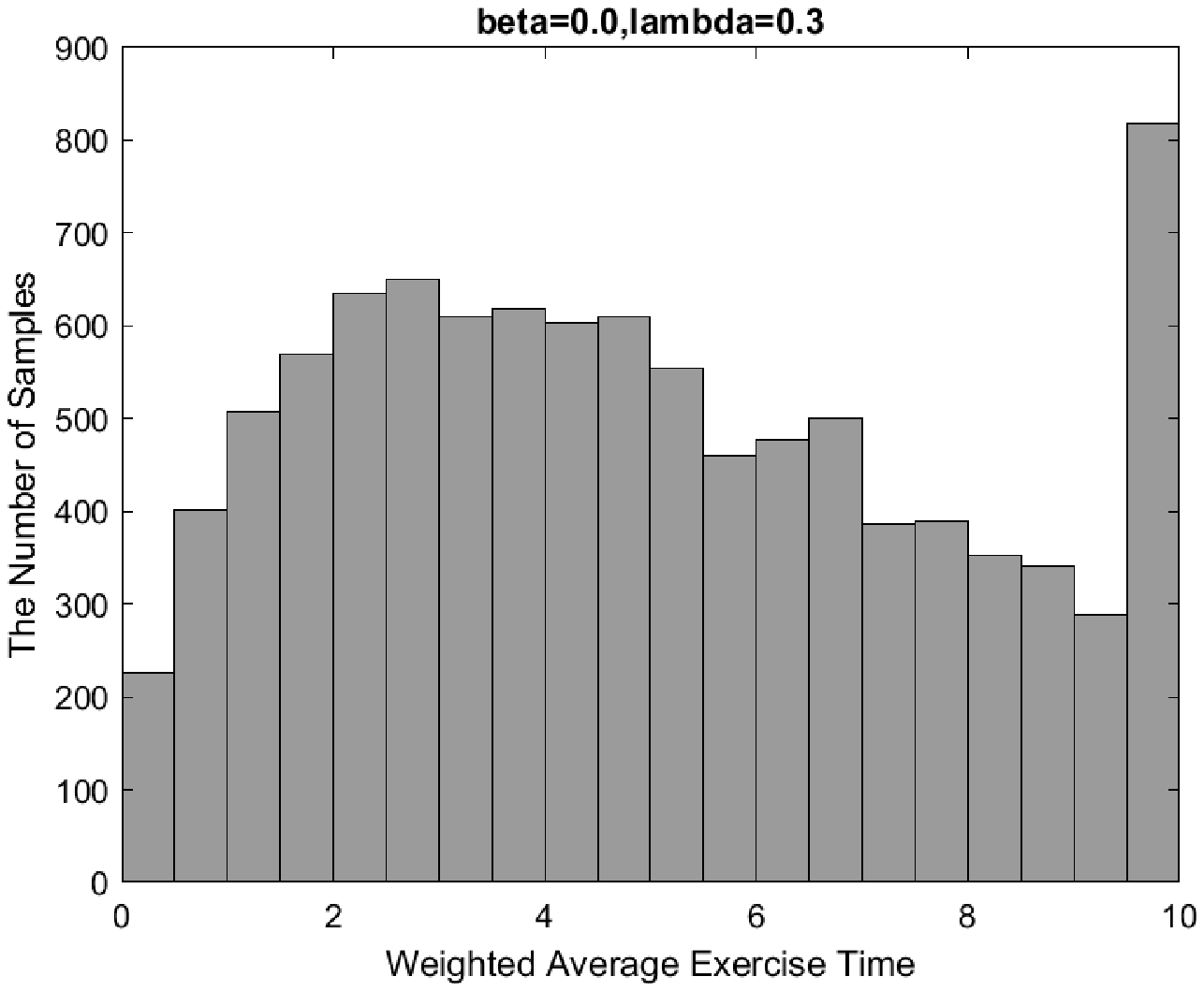}
}
\subfigure[ ]
{\includegraphics[trim={1cm 0cm 1cm 0.1cm},clip,width=2.8in, angle=0]{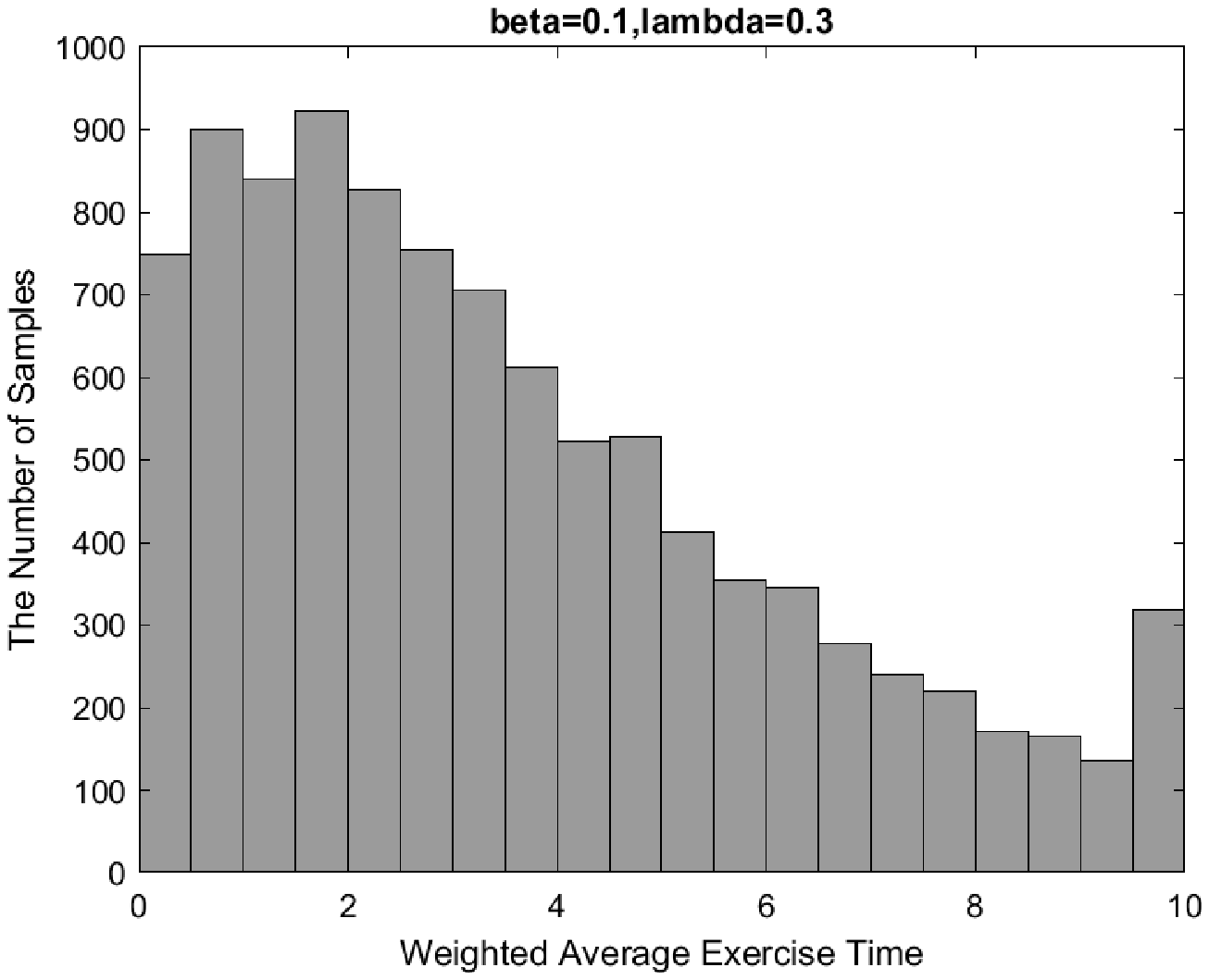}
}
\subfigure[ ]
{\includegraphics[trim={1cm 0cm 1cm 0.1cm},clip,width=2.8in, angle=0]{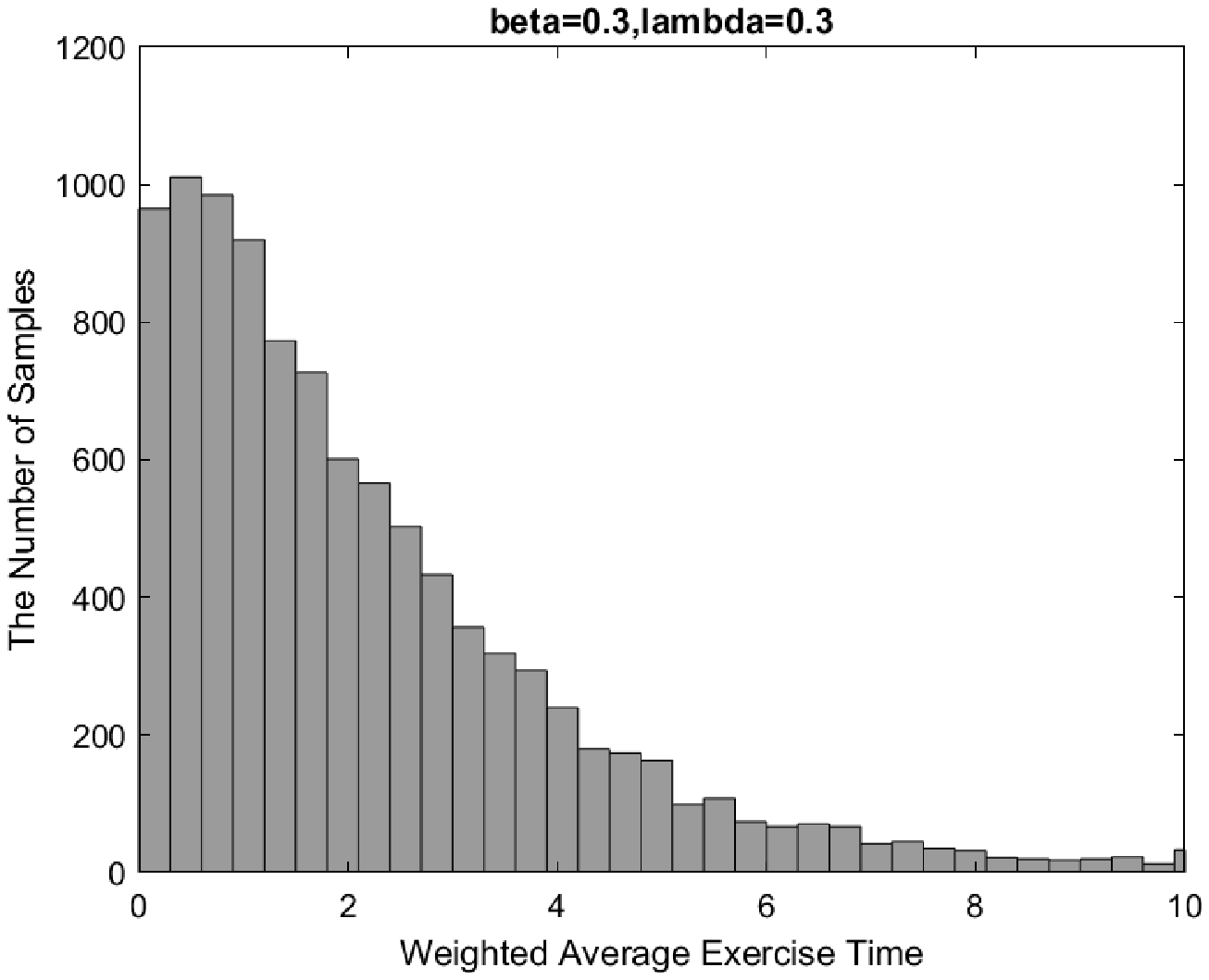}
}
\subfigure[ ]
{\includegraphics[trim={1cm 0cm 1cm 0.1cm},clip,width=2.8in, angle=0]{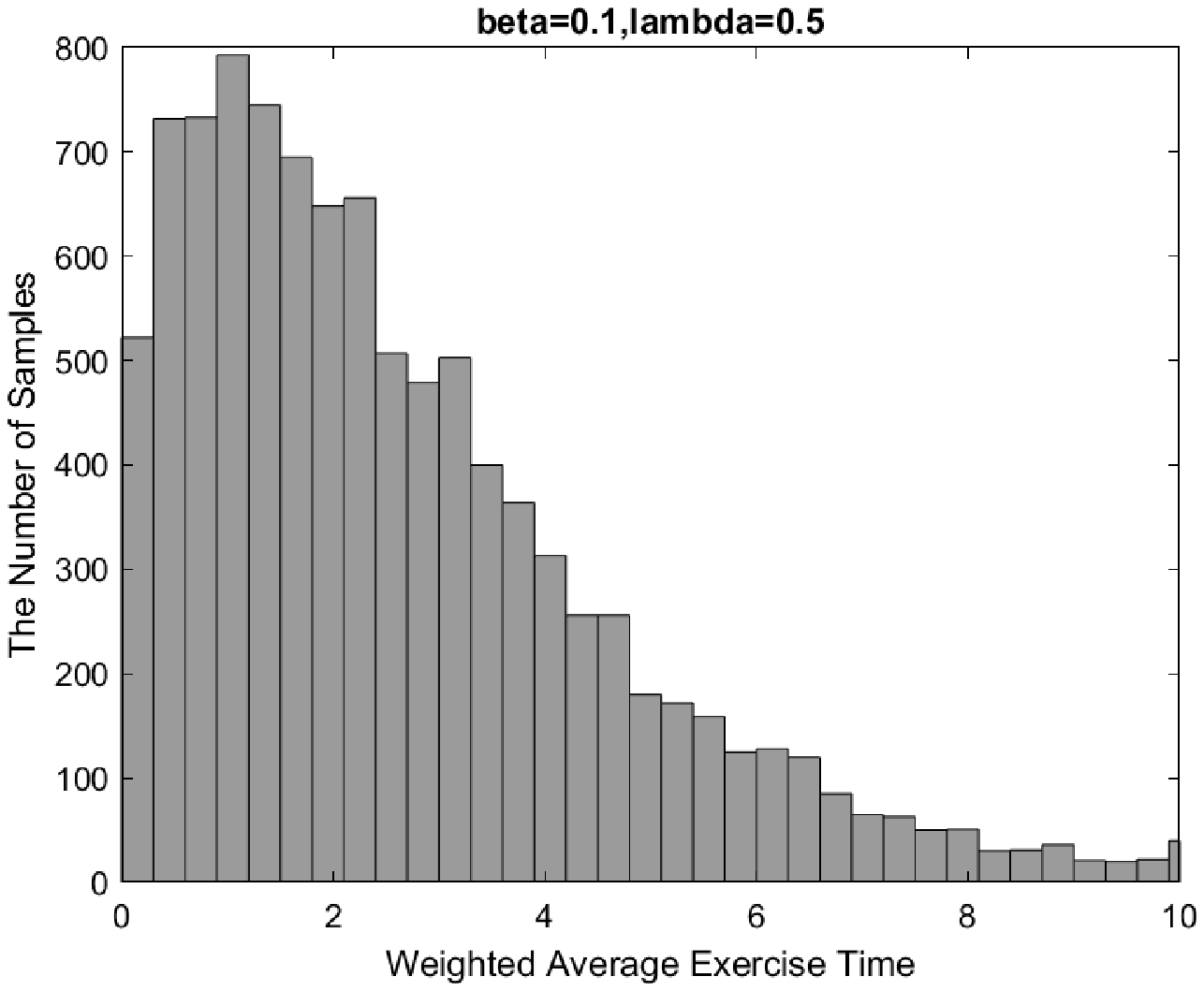}
}
\caption{Histograms of weighted average exercise times, as defined in \eqref{AverageExerciseTime}, based on 10,000 simulated exercise  processes for 20 vested ESOs with  a 10-year maturity.   Panels have different rates of job termination  $\beta$ and exercise intensity $\lambda$. (a): $\beta=0$, $\lambda=0.3$; (b): $\beta=0.1$, $\lambda=0.3$; (c): $\beta=0.3$, $\lambda=0.3$; (d): $\beta=0.1$, $\lambda=0.5$. }\label{AverageExercise}
\end{figure}

%\begin{remark}
%Since the continuous distribution could be approximated by the
%discrete distribution, we only consider the discrete distribution in this paper.
%\end{remark}

\clearpage
\subsection{PDEs for ESO Valuation}\label{sect-ESOPDE}
To value ESOs, we consider a risk-neutral pricing measure $\Q$ for all stochastic processes and random variables in our model. We model the firm's stock price process  $(S_t)_{t\ge 0}$ by a geometric Brownian motion
\begin{equation}\label{S}
dS_t = (r-q) S_t \,dt + \sigma S_t \, dW_t,
\end{equation}
where the positive constants  $r$, $q$ and $\sigma$ are the interest rate, dividend rate,  and volatility parameter  respectively, and $W$ is a standard Brownian motion under $\Q$, independent of the exponentially-distributed job termination times $\zeta$ and $\xi$.  Our default assumption for the employee's exercise intensity is that it is a deterministic  function of time, denoted by $\lambda(t)$.  We will discuss the case with a stochastic exercise intensity in Section \ref{sect-stochastic}. 

At any time $t\in [t_v, T]$, the ESO is vested. The   vested ESO cost functions $C^{(m)}(t,s)$, for $m=1, 2, \ldots, M$, where $m$ is the number of options currently held, are given by the risk-neutral expectation of discounted future ESO payoffs provided that the employee has not left the firm.  
\begin{align}
C^{(m)}(t,s) =& \E\bigg\{ \int_{t}^{T\wedge \xi} e^{-r (u-t)} (S_u -K)^+ dL_u \notag \\
&~+ e^{-r(T\wedge \xi-t)} (M- L_{T\wedge \xi}) (S_{T\wedge \xi}-K)^+\,|\, S_t = s, L_t = M-m \bigg\}\notag\\
=&\E\bigg\{ \int_{t}^{T} e^{-(r+\beta) (u-t)} (S_u -K)^+ dL_u+ e^{-(r+\beta)(T-t)} (M- L_T) (S_T-K)^+ \notag \\
&~+\int_t^T\beta e^{-(r+\beta)(v-t)}(M-L_v)(S_v-K)^+dv\,|\, S_t = s, L_t = M-m \bigg\},\label{Cmts}
\end{align} for $m=1, 2, \ldots, M$, and  $(t,s)\in[t_v, T]\!\times\!\R_+$.  

Next, we  define the infinitesimal generator associated with the stock price process $S$ by 
\begin{equation}
\mathcal{L}\,\cdot=(r-q)s\partial_s\cdot\,+\frac{\sigma^2s^2}{2}\partial_{ss}\,\cdot\,.\label{L}
\end{equation}
We determine the vested ESO costs by solving the following system of PDEs. 
\begin{align}\label{MultiplePDE}
-(r+\lambda(t)+\beta)C^{(m)}+C^{(m)}_t+\mathcal{L}C^{(m)}+ \lambda(t)\sum_{z=1}^{m-1} p_{m,z}C^{(m-z)} +  \left(\lambda(t)\bar p_m+m\beta\right) (s-K)^+=0,
\end{align} for $(t,s)\in[t_v, T]\!\times\!\R_+$ and $m=1, 2, \ldots, M$. Here, $\bar p_m$ is the expected number of options exercised and $p_{m,z}$ is the probability of exercising $z$ options with $m$ options left. The terminal condition is $C^{(m)}(T,s) =  m(s-K)^+$  for $s\in\R_+$.

 During the vesting period $[0,t_v)$, the ESO is unvested and is subject to forfeiture if the employee leaves the firm. We denote the cost of $m$  units of unvested ESO by $\tilde C^{(m)}(t,s)$. Since holding an unvested ESO effectively entitles the holder to obtain a vested ESO  at time $t_v$ provided the holder is still with the firm. If the ESO holder leaves the firm at any time $t\in[0,t_v)$, the unvested ESO cost is zero.  Otherwise, given that $\zeta >t$, the (pre-departure) unvested ESO cost  is 
  \begin{align}
\tilde{C}^{(m)}(t,s)& = \E\left\{e^{-r(t_v-t)} C^{(m)}(t_v,S_{t_v})\textbf{1}_{\{\zeta\ge t_v\}}| S_t=s\right\}\notag\\
&=\E\left\{e^{-(r+\alpha)(t_v-t)} C^{(m)}(t_v,S_{t_v})|S_t=s\right\}.\label{unvestedESOcost}
\end{align}
To  determine the unvested ESO cost, we  solve the PDE problem
\begin{equation}\label{unvestedESO}
\begin{aligned}
-(r+\alpha)\tilde{C}^{(m)}+\tilde{C}^{(m)}_t+\mathcal{L}\tilde{C}^{(m)}=0,&&\mbox{ for }(t,s)\in[0,t_v)\times\R_+,\\
\tilde{C}^{(m)}(t_v,s)=C^{(m)}(t_v,s),&& \mbox{ for }s\in\R^+.
\end{aligned}
\end{equation}
Here, $C^{(m)}(t_v,s)$ is the  vested ESO cost evaluated at time $t_v$.

\section{Numerical Methods and Implementation}\label{sect-num}
In this section, we present two numerical methods to solve PDE (\ref{MultiplePDE}). We first discuss the application of fast Fourier transform (FFT) to ESO valuation, followed by the finite difference method (FDM). The results from the two methods are compared in Section \ref{sect-num3}.

\subsection{Fast Fourier Transform}\label{FFT}

%\subsubsection{Vested ESO}\label{vestedESOFourier}

We first consider the vested ESO ($t\in [t_v,T]$). Let $x$ such that $s=Ke^x$, and define the function
\begin{equation}
f^{(m)}(t,x)=C^{(m)}(t,Ke^x), \quad  (t,x)\in[t_v,T]\times\R,
\end{equation}
for each $m=1,\ldots,M$. The PDE for $f^{(m)}(t,x)$ is given by
\begin{equation}
-(r+\lambda(t)+\beta)f^{(m)}+f^{(m)}_t+\widetilde{\mathcal{L}} f^{(m)}+ \lambda(t)\sum_{z=1}^{m-1} p_{m,z}f^{(m-z)} +  \left(\lambda(t)\bar p_m+m\beta\right) (Ke^x-K)^+=0,\label{FourierPDE}
\end{equation}
where 
\begin{equation}
\widetilde{\mathcal{L}} \,\cdot = (r-q-\frac{\sigma^2}{2})\partial_x\cdot+\frac{\sigma^2}{2}\partial_{xx}\cdot\,. \label{Lf}
\end{equation}
The  terminal condition is $f^{(m)}(T,x) =  m(Ke^x-K)^+$, for $x\in\R$.
%for $m= 1,2,\ldots,M$, and $(t,x)\in[t_v,T]\times\R$, 

The Fourier transform of $f^{(m)}(t,x)$ is defined by
\begin{equation}
\mathcal{F}[f^{(m)}](t,\omega)=\int_{-\infty}^{\infty}f^{(m)}(t,x)e^{-i\omega x}dx,
\label{FourierTransform}
\end{equation}
for $m=1,\ldots,M$, with angular frequency $\omega$ in radians per second.  Applying Fourier transform  to PDE (\ref{FourierPDE}), we  obtain an ODE for $\mathcal{F}[f^{(m)}](t,\omega)$, a function of time $t$ parametrized by $\omega$, for each $m=1,\ldots,M$. Precisely, we have  
\begin{equation}
\frac{d}{dt}\mathcal{F}[f^{(m)}](t,\omega)=h(t,\omega)\mathcal{F}[f^{(m)}](t,\omega)+\psi^{(m)}(t,\omega),
\end{equation}
where
\begin{align}
h(t,\omega)&=r+\lambda(t)+\beta-i\omega (r-q-\frac{\sigma^2}{2})+\omega ^2\frac{\sigma^2}{2},\\
\psi^{(m)}(t,\omega)&=-\lambda(t)\sum_{z=1}^{m-1} p_{m,z}\mathcal{F}[f^{(m-z)}](t,\omega)-\left(\lambda(t)\bar p_m+m\beta \right)\varphi(\omega) ,\\
\varphi(\omega) &= \mathcal{F}[(Ke^x-K)^+](\omega),\label{varphi}
\end{align} with the terminal condition $\mathcal{F}[f^{(m)}](T,\omega)=m\varphi(\omega)$. Solving the ODE, we obtain
\begin{equation}
\mathcal{F}[f^{(m)}](t,\omega)=e^{-\int_t^T h(s,\omega)ds}\mathcal{F}[f^{(m)}](T,\omega)-\int_t^Te^{-\int_t^uh(s,\omega)ds}\psi^{(m)}(u,\omega)du.\label{FFTSolution}
\end{equation}
%for $m=1,\ldots,M$ and $(t,\omega)\in(t_v,T)\times\mathbb{R}$.
Accordingly, we can recover  the vested ESO cost function   by inverse Fourier transform:
\begin{equation}
f^{(m)}(t,x)=\mathcal{F}^{-1}[\mathcal{F}[f^{(m)}]](t,x).\end{equation}
for every  $m=1,\ldots,M$, and $(t,x)\in(t_v,T)\times\R.$

In the literature,   \cite{LeungWan} apply a Fourier time-stepping (FST) method it to compute the cost of an American-style ESO when the company stock is driven by a Levy process. This FST method has been  applied more broadly by \cite{JacksonJaimungalSurkov2008} to solve partial-integro differential equations (PIDEs) that arise in options pricing problems.  
 
\begin{remark}
If $\lambda$ is a constant, then   the Fourier transform in  (\ref{FFTSolution}) can be simplified as
%could be simplified as
%\begin{equation}
%\mathcal{F}[f^{(m)}](t,\omega)=e^{-(T-t)h(\omega)}\mathcal{F}[f^{(m)}](T,\omega)-\int_t^Te^{-(u-t)h(\omega)}\psi^{(m)}(u,\omega)du,
%\end{equation}
%{for} $(t,\omega)\in(t_v,T)\times\mathbb{R}$, where
\begin{equation}
\mathcal{F}[f^{(m)}](t,\omega)=\sum_{k=0}^{m-1}F^{(m)}_k(\omega)(T-t)^ke^{-(T-t)h(\omega)}+F^{(m)}(\omega),
\end{equation}
where
\begin{align}
F^{(m)}(\omega)&=\frac{1}{h(\omega)}\bigg(\lambda\sum_{z=1}^{m-1}p_{m,z}F^{(m-z)}(\omega)+(\lambda\bar p_m+m\beta)\varphi(\omega)\bigg),  \label{ceof1}\\
F^{(m)}_k(\omega)&=\frac{\lambda}{k}\sum_{z=1}^{m-k}p_{m,z}F^{(m-z)}_{k-1}(\omega),\quad k=1,2,
\ldots,m-1,  \label{ceof2}\\
F^{(m)}_0(\omega)&=\mathcal{F}[f^{(m)}](T,\omega)-\frac{1}{h(\omega)}\bigg(\lambda\sum_{z=1}^{m-1}p_{m,z}F^{(m-z)}(\omega)+(\lambda\bar p_m+m\beta)\varphi(\omega)\bigg), \label{ceof3}\\
h(\omega)&=r+\lambda +\beta-i\omega (r-q-\frac{\sigma^2}{2})+\omega ^2\frac{\sigma^2}{2}. \label{ceof4}
\end{align}
%for $m=1,2,\ldots,M$, and $\omega\in\mathbb{R}$.
In \eqref{ceof1} and \eqref{ceof3}, $\varphi(\omega)$ is defined in (\ref{varphi}).
\end{remark}

For numerical implementation, we work with a finite domain $[t_v,T]\times[x_{min},x_{max}]$ with uniform discretization of lengths $\delta t=(T-t_v)/N_t$ and $\delta x=(x_{max}-x_{min})/(N_x-1)$ in the time-space dimensions.  We set $\delta t=0.01$, $x_{min}=-10$, $x_{max}=10$ and $N_x=2^{12}$. Similarly, we discrete the finite frequency space $[\omega_{min},\omega_{max}]$ with uniform grid size of 
$\delta\omega$,  where we apply the Nyquist critical frequency that $\omega_{max}=\pi/\delta x$ and $\delta\omega=2\omega_{max}/N_x$.
For $j= 0,\ldots,N_t$, and $k=0,\ldots,N_x-1$, we denote $t_j=t_v+j\delta t$, $x_k=x_{min}+k\delta x$, and 
\begin{equation}
\omega_k=\left\{
\begin{aligned}
&k\delta\omega,&0\le k \le N_x/2\,,\\
&k\delta\omega-2\omega_{max},&N_x/2+1\le k \le N_x-1]\,.
\end{aligned}
\right.
\end{equation}
Then we numerically compute the discrete Fourier transform 
\begin{align}
\mathcal{F}[f](t_j,\omega_k)&\approx\sum_{n=0}^{N_x-1}f(t_j,x_n)
e^{-i\omega_kx_n}\delta x=\phi_k\sum_{n=0}^{N_x-1}f(t_j,x_n)e^{-i2\pi kn/N_x},\label{dft}
%f(t_j,x_n)&\approx\sum_{k=0}^{N_x-1}\mathcal{F}[f](t_j,\omega_k)e^{-i\omega_kx_n}\delta \omega=\psi_n^{-1}\sum_{k=0}^{N_x-1}\mathcal{F}[f](t_j,\omega_k)e^{i2\pi k n/N_x},
\end{align}
with $\phi_k=e^{-i\omega_k x_{min}}\delta x$. In (\ref{dft}), we evaluate the sum $\sum_{n=0}^{N_x-1}f(t_j,x_n)e^{-i2\pi k n/N_x}$ by applying the standard  fast Fourier transform (FFT) algorithm. The corresponding Fourier inversion is conducted by inverse FFT, yielding the vested ESO cost $f(t_j,x_n)$. Note that the coefficient $\phi_k$ will be cancelled in the process.

As for the unvested ESO, we define the associated cost function
\begin{equation}
\tilde f^{(m)}(t,x)=\tilde C^{(m)}(t,Ke^x),
\end{equation}
for each $m=1,\ldots,M$. From PDE (\ref{unvestedESO}), we derive the PDE for $\tilde f^{(m)}(t,x)$
\begin{equation}\label{FourierPDE_unvested}
-(r+\alpha)\tilde f^{(m)}+\tilde f^{(m)}_t+\widetilde{\mathcal{L}}\tilde f^{(m)}=0,
\end{equation}
for $(t,x)\in[0,t_v)\times\R$, with the terminal condition $\tilde f^{(m)}(t_v,x) =  f^{(m)}(t_v,x)$, for $x\in\R$. As we can see, once the vested ESO cost is computed, it determines the terminal condition for the unvested ESO problem.

%In PDE (\ref{unvestedESO}), we shall see the computation for the unvested ESO ($\tilde C^{(m)}(t,s)$) depends on its terminal condition, which is the corresponding vested ESO. Thus, we need to evaluate the vested ESO cost ($C^{(m)}(t,s)$)  at first. Since the computation for $C^{(m)}(t,s)$ depends on $C^{(i)}(t,s)$, where $i=1,2,\ldots,m-1$, see PDE (\ref{MultiplePDE}), we compute the numerical solution for $C^{(1)}(t,s)$ at first, then $C^{(2)}(t,s)$, $C^{(3)}(t,s)$, etc.

 %\begin{equation}
%\hat{\tilde {f}}^{(m)}(t,\omega)=\int_{-\infty}^{\infty}\tilde f^{(m)}(t,x)e^{-i\omega x}dx\quad \mbox{for }(t,\omega)\in[0,t_v]\times\R
%\end{equation}
%denote the Fourier transform of $\tilde f^{(m)}(t,x)$.

Applying Fourier transform   to (\ref{FourierPDE_unvested}), we can derive the ODE for $\mathcal{F}[\tilde{f}^{(m)}](t,\omega)$,
\begin{equation}
\frac{d}{dt}\mathcal{F}[\tilde{f}^{(m)}](t,\omega)=\tilde h(\omega)\mathcal{F}[\tilde{f}^{(m)}](t,\omega),
\end{equation}
where
\begin{equation}
\tilde h(\omega)=r+\alpha-i\omega(r-q-\frac{\sigma^2}{2})+\omega^2\frac{\sigma^2}{2},
\end{equation}
for $(t,\omega)\in[0,t_v)\times\mathbb{R}$, with the terminal condition $\mathcal{F}[\tilde{f}^{(m)}](t_v,\omega)=\mathcal{F}[f^{(m)}](t_v,\omega)$. We solve the ODE to get  
\begin{equation}
\mathcal{F}[\tilde{f}^{(m)}](t,\omega)=e^{-\tilde h(\omega)(t_v-t)}\mathcal{F}[\tilde{f}^{(m)}](t_v,\omega).
\end{equation}
In turn, we apply inverse Fourier transform to recover the unvested ESO cost:
\begin{equation}
\tilde C^{(m)}(t,Ke^x)=\tilde f^{(m)}(t,x)=\mathcal{F}^{-1}[\mathcal{F}[\tilde f^{(m)}]](t,x),
\end{equation}
for $(t,x)\in[0,t_v)\times\R$. Again, we apply FFT to numerically compute the Fourier transform and use inverse FFT to recover the cost function. 

\subsection{Finite Difference Method}\label{FDM}
For comparison, we also compute the ESO costs using a finite difference method. Specifically, we apply  the Crank-Nicolson method   on a uniform grid. Here we provide an outline with focus on the boundary conditions for our application. For more details, we refer to \cite{WilmottHowisonDewynne1995}, among other references. 

%Then, we compute numerical solution for $C^{(2)}$ based on the numerical solution for $C^{(1)}$. Then we compute the numerical solution for $C^{(3)}$ based on the numerical solution for $C^{(1)}$ and $C^{(2)}$, etc. Basically, we compute the numerical solution for $C^{(m)}$ based on the numerical solution for $C^{(i)}$, where $i=1,2,\ldots,m-1$.

  As for grid settings, We restrict the domain $[t_v,T]\times \R_+$ to a finite domain $\mathcal{D}=\{(t,s)|t_v\le t\le T, 0 \le s \le S_*\}$, where $S_*$ must be relatively very large such that if the current stock price $S_t=S_*$, then the stock price will be larger than the strike price $K$ over $[t,T]$ with great probability. 

To determine the boundary condition at $s=S_*$, we introduce a new function 
\begin{equation}
\begin{aligned}
\bar C^{(m)}(t,s)&=\E\bigg\{ \int_{t}^{T} e^{-(r+\beta) (u-t)} (S_u -K) dL_u+ e^{-(r+\beta)(T-t)} (M- L_T) (S_T-K) \notag \\
&\quad+\int_t^T\beta e^{-(r+\beta)(v-t)}(M-L_v)(S_v-K)dv\,|\, S_t = s, L_t = M-m \bigg\}.
%&=\mbox{need to be determined}.
\end{aligned}
\end{equation}
for $m=1,\ldots,M$. When $s=S_*$, we see that $C^{(m)}(t,s)\approx\bar C^{(m)}(t,s)$. Thus, we can set the boundary condition at $s=S_*$ to be $C^{(m)}(t,S_*)=\bar C^{(m)}(t,S_*)$. By Feynman-Kac formula,   $\bar C^{(m)}(t,s)$ satisfies the PDE
\begin{equation}
-(r+\lambda(t)+\beta)\bar C^{(m)}+\bar C^{(m)}_t+\mathcal{L}\bar C^{(m)}+ \lambda(t)\sum_{z=1}^{m-1} p_{m,z}\bar C^{(m-z)}
 +  \left(\lambda(t)\bar p_m+m\beta\right) (s-K)=0,
\end{equation}
for $m=1,\ldots,M$, and $(t,s)\in(t_v,T)\times\R_+$, with terminal condition $\bar C^{(m)}(T,s) =  m(s-K)$, for $s\in\R_+$.
Then, $\bar C^{(m)}(t,s)$ has the ansatz solution
\begin{equation}
\bar C^{(m)}(t,s)=A_m(t)s-B_m(t)K,
\end{equation}
where $A_m(t)$ and $B_m(t)$ satisfy the pair of ODEs respectively,
\begin{equation}
\label{eqnAB}
\begin{aligned}
-(q+\lambda(t)+\beta)A_m+A_m'+ \lambda(t)\sum_{z=1}^{m-1} p_{m,z}A_{m-z} +  \left(\lambda(t)\bar p_m+m\beta\right)=0,\\
-(r+\lambda(t)+\beta)B_m+B_m'+ \lambda(t)\sum_{z=1}^{m-1} p_{m,z}B_{m-z} +  \left(\lambda(t)\bar p_m+m\beta\right)=0,\\
\end{aligned}
\end{equation}
for $m=1,\ldots,M$, and $t\in(t_v,T)$, with the terminal condition $B_m(T)=A_m(T) =  m$, for $m=1,\ldots,M$.
We can solve the ODEs (\ref{eqnAB}) analytically, or numerically solve it using the backward Euler method.
%That are
%\begin{equation}
%A_m'(t)\approx \frac{A_m(t+\delta t)-A_m(t)}{\delta t}
%\end{equation}
%and
%\begin{equation}
%B_m'(t)\approx \frac{B_m(t+\delta t)-B_m(t)}{\delta t}.
%\end{equation}

Next, we discrete the domain $\mathcal{D}$ with uniform grid size of $\delta t= (T-t_v)/M_0$ and $\delta S=S_*/N_0$.
Then, we apply $C^{(m)}_{i,j}$ to denote discrete approximations of $C^{(m)}(t_i,s_j)$ where $t_i=t_v+i\delta t$ and $s_j=j\delta S$. The Crank-Nicolson method is applied to solve the PDEs satisfied by $C^{(m)}$, for $m=1,\ldots, M$. Working backward in time, we obtain the vested ESO costs at time $t_v$, which become the   terminal condition values for the unvested ESO valuation problem. For the unvested ESO cost, we restrict the domain $[0,t_v]\times \R_+$ to the finite domain $\tilde{\mathcal{D}}=\{(t,s)|0\le t\le t_v, 0 \le s \le S_*\}$, where $S_*$ is relatively very large such that
\begin{align}
\tilde C^{(m)}(t,S_*)&=\E\left\{e^{-(r+\alpha)(t_v-t)} C^{(m)}(t_v,S_{t_v})|S_t=S_*\right\}\\
&\approx\E\left\{e^{-(r+\alpha)(t_v-t)} (A_m(T-t_v)S_{t_v}-B_m(T-t_v)K)|S_t=S_*\right\}\\
&=e^{-(q+\alpha)(t_v-t)}A_m(T-t_v)S_*-e^{-(r+\alpha)(t_v-t)}B_m(T-t_v)K.
\end{align}
We again apply the  Crank-Nicolson method solve the PDEs satisfied by $\tilde{C}^{(m)}(t,s)$,  for $m=1,\ldots, M$.

\subsection{Numerical Examples}\label{sect-num3}
Using both FFT and FDM we compute different ESO costs by varying the vesting period  $t_v$, job termination rate $\alpha$ and $\beta$, as well as  exercise intensity $\lambda$. In Table \ref{ComparisonI}, we present the ESO costs  and compare   the two numerical methods.  It is well known that the call option value is increasing with respect to its maturity. In a similar spirit if the employee tends to exercise the ESO earlier, then a smaller ESO cost is expected. As we can see in Table \ref{ComparisonI},  the ESO cost  decreases as  exercise intensity $\lambda$  increases, or as job termination rate $\alpha$ or $\beta$ increases, holding other things constant.  On the other hand, the effect of vesting period is not monotone.  In a scenario with a high job termination $\alpha$ during the vesting period, the employee is very likely to leave the firm while the options are unvested, leading to a zero payoff. Consequently, the ESO cost is decreasing with respect to $t_v$. This corresponds to the case with $\alpha=1$ in Table \ref{ComparisonI}. However, if $\alpha$ is small, then the employee is unlikely to leave the firm and lose the options during the vesting period. Therefore,   a longer vesting period would effectively make the employee hold  the options for a longer period of time, delaying the exercise. As a result,  the ESO cost is increasing with respect to $t_v$, which is shown in other cases in Table \ref{ComparisonI}.

\begin{table}[h]
  \centering
  \begin{tabular}{cc|cc|cc|cc}
     \hline
     \multicolumn{2}{c}{\multirow{2}{*}{Parameters}} &\multicolumn{2}{|c|}{$t_v=0$}&\multicolumn{2}{|c}{$t_v=2$}&\multicolumn{2}{|c}{$t_v=4$}\\
     \cline{3-8}
      & & FDM & FFT & FDM & FFT & FDM & FFT \\
      \hline
      \hline
%&\\
%    $\alpha=0.1$&\\
%\hline
      \multirow{2}{*}{$\alpha=0.1,\beta=0$} & $\lambda=1$ & 5.4729 &5.4753 & 7.8399 & 7.8405 & 8.2845& 8.2849\\
      & $\lambda=2$&3.7067 &3.7101 &6.9164&6.9170&7.7054 &7.7058\\
     \hline
      \multirow{2}{*}{$\alpha=0.1,\beta=1$} & $\lambda=1$ &3.2483 &3.2522 & 6.7063 &6.7069 &7.5746 & 7.5750\\
      & $\lambda=2$&2.7024 & 2.7069&6.4655 & 6.4661 &7.4253& 7.4257\\
     \hline
    \hline
%&\\
%$\beta=0.1$&\\
%\hline
      \multirow{2}{*}{$\alpha=0,\beta=0.1$} & $\lambda=1$ &5.0603&5.0629& 9.3022  &9.3031 & 12.1510& 12.1517\\
      & $\lambda=2$&3.5595&3.5630&8.3622 &8.3631 &11.4298 &11.4306\\
     \hline
      \multirow{2}{*}{$\alpha=1,\beta=0.1$} & $\lambda=1$ &5.0603&    5.0629&1.2579 &1.2590 &0.2219 & 0.2226\\
      & $\lambda=2$&3.5595  &  3.5630&1.1310 & 1.1318& 0.2087& 0.2094\\
     \hline
   \end{tabular}
  \caption{Vested and unvested ESO costs  under different exercise  intensities $\lambda$ and different job termination rates $\alpha$ and $\beta$,  computed using FFT and FDM  for comparison. Common Parameters: $S_0=K=10$, $r=5\%$, $q=1.5\%$, $\sigma=20\%$, $p_{m,z}=1/m$, $M=5$ and $T=10$. In FDM: $S_*=30$, $\delta S=0.1$, $\delta t=0.1$. In FFT: $N_x=2^{12}$, $x_{min}=-10$ and $x_{max}=10$.}\label{ComparisonI}
\end{table}
\vspace{10pt}

  In Figure \ref{lambda-ESO}, we plot the ESO cost as a function of the exercise intensity $\lambda$ for $T=5$, 8 and 10. It shows that the ESO is decreasing and convex with respect to $\lambda$. An employee with a high exercise intensity tends to exercise the ESOs earlier than those with a lower  exercise intensity. Since the call option value increases with maturity, exercising the ESO  earlier will result in a lower cost.  As the exercise intensity increases from zero, the ESO cost tends to decrease faster than when the exercise intensity is higher. Moreover, Figure \ref{lambda-ESO} also shows that as $\lambda$ increases, the ESO costs associated with different maturities $T=5$, 8 and 10 get close to each other. The intuition is that  when $\lambda$ is large, the options  will be exercised very early and the maturity will not have a significant impact on the option values.

On the right panel of Figure \ref{lambda-ESO}, we plot the option value as a function of stock price $S_0$ with exercise intensity  $\lambda = 0$, 1 or $5$.  It shows that, as $\lambda$ increases from 0 to 1, the option value decreases rapidly. When the  exercise intensity is very high, i.e. $\lambda=5$ in the figure, there is a  high chance  of immediate exercise, so the ESO value is seen to be very close to the ESO payoff   $(S_0-K)^+$.  

Next, we consider  the effect of the total number of ESOs granted. Intuitively we expect the total cost  to increase as the number of options $M$ increases, but the effect is far from linear.  In Figure \ref{M-AverageExercise}  we see that the average per-unit cost and average time to exercise are increasing as $M$ increases.  In other words, under the assumption that the ESOs will be exercised gradually, a larger ESO grant has an indirect effect of delaying exercises, and thus leading to higher ESO costs. 
  The increasing trends hold  for different exercise intensities, but the rate of increase diminishes significantly for large $M$. Also, the higher the exercise intensity, the lower the per-unit cost and shorter averaged time to exercise.

 \begin{figure}[h]
  \centering
    \includegraphics[width=2.9in]{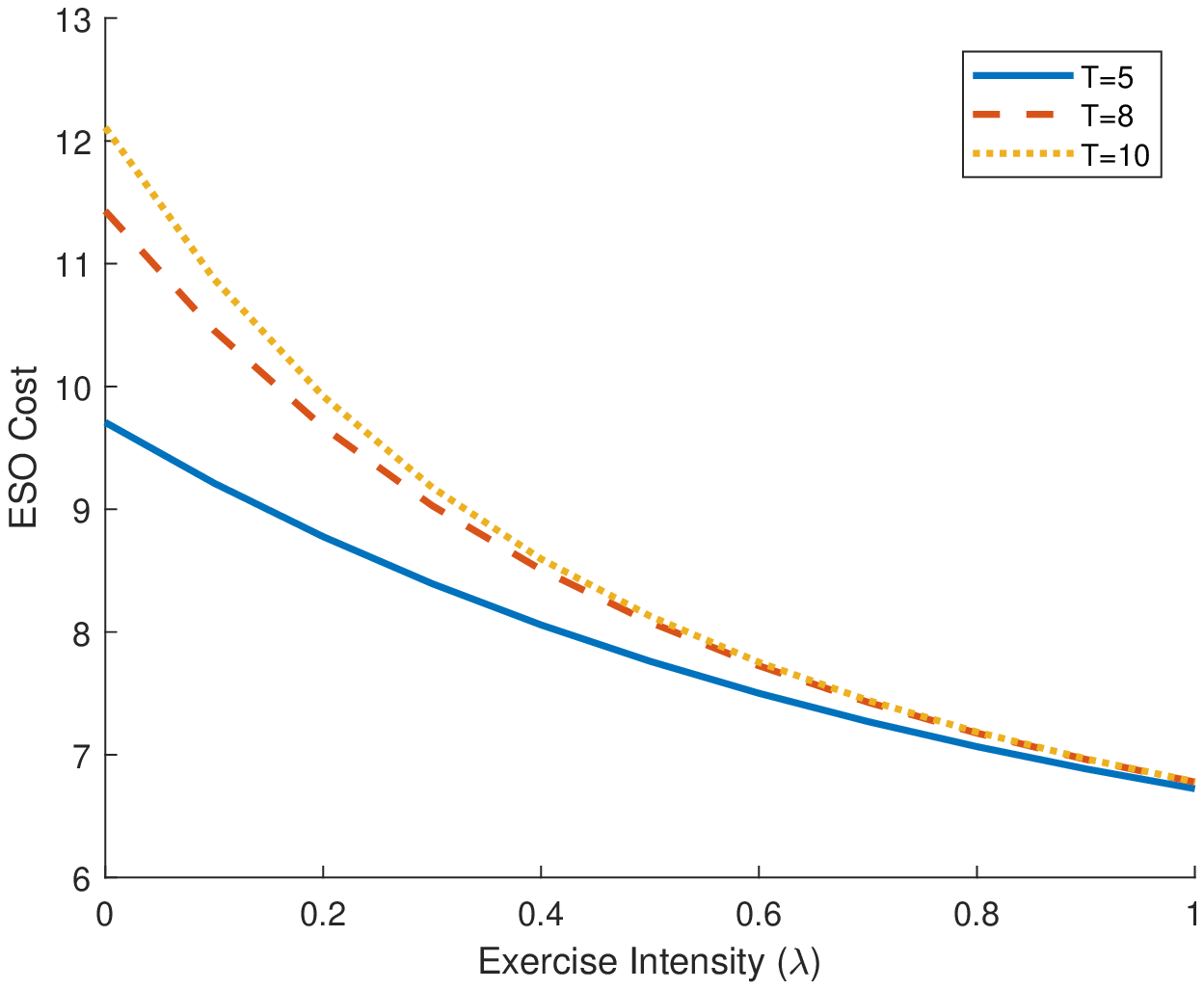}  \includegraphics[width=2.9in]{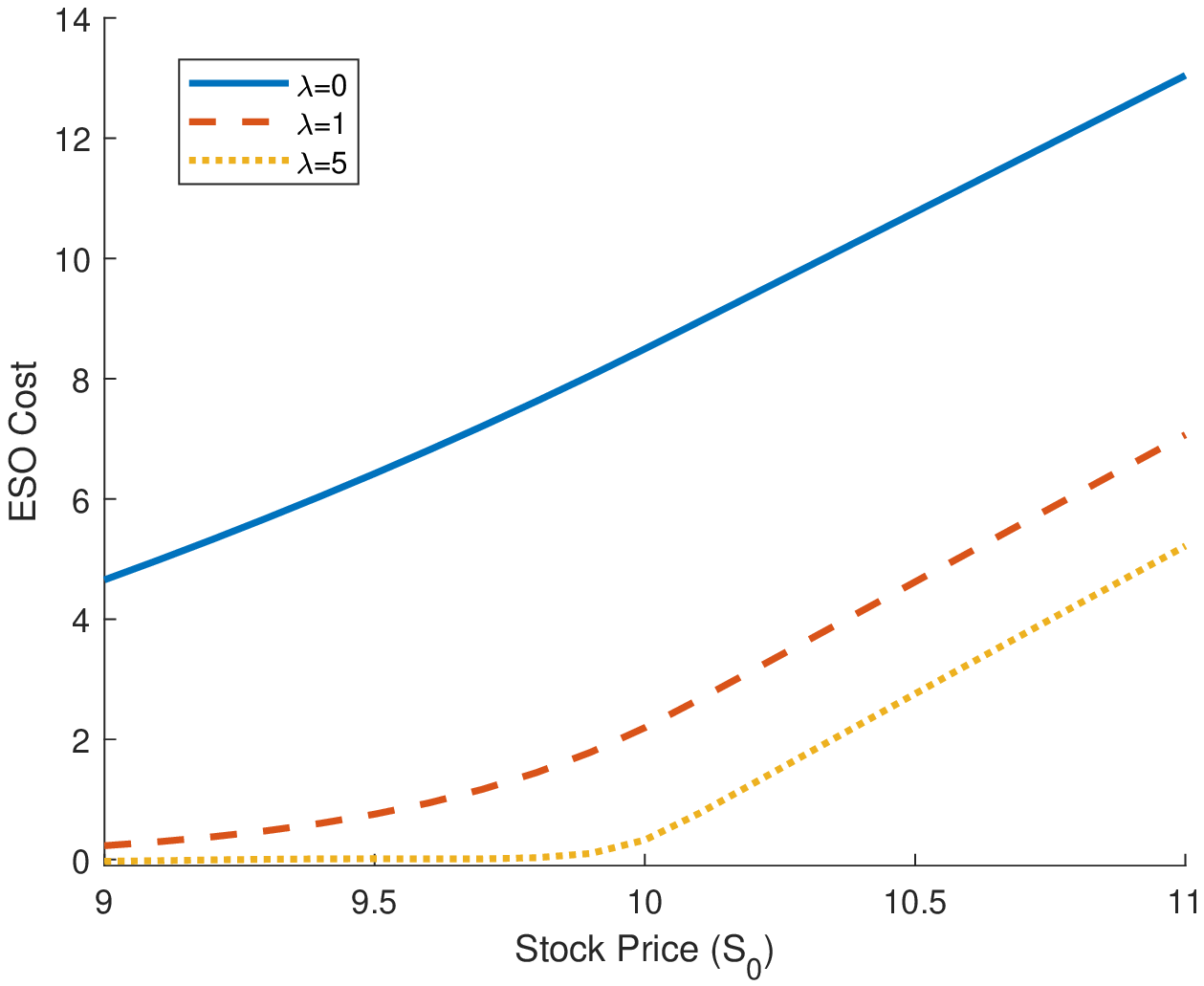}\\
  \caption{Left:   ESO cost as a function of employee exercise intensity  $\lambda$ when the maturity $T= 5, 8$ or $10$.  Right: ESO cost as a function of initial stock price $S_0$ with $\lambda = 0, 1$, or $5$.  Parameters: $K=10$, $r=5\%$, $q=1.5\%$, $\sigma=20\%$, $p_{m,z}=1/m$, $M=5$, $T=10$, $t_v=0$ and $\beta=0.1$. In FFT: $N_x=2^{12}$, $x_{min}=-10$, $x_{max}=10$.}\label{lambda-ESO}
\end{figure}
\begin{figure}[h]
  \centering
  \includegraphics[width=2.9in]{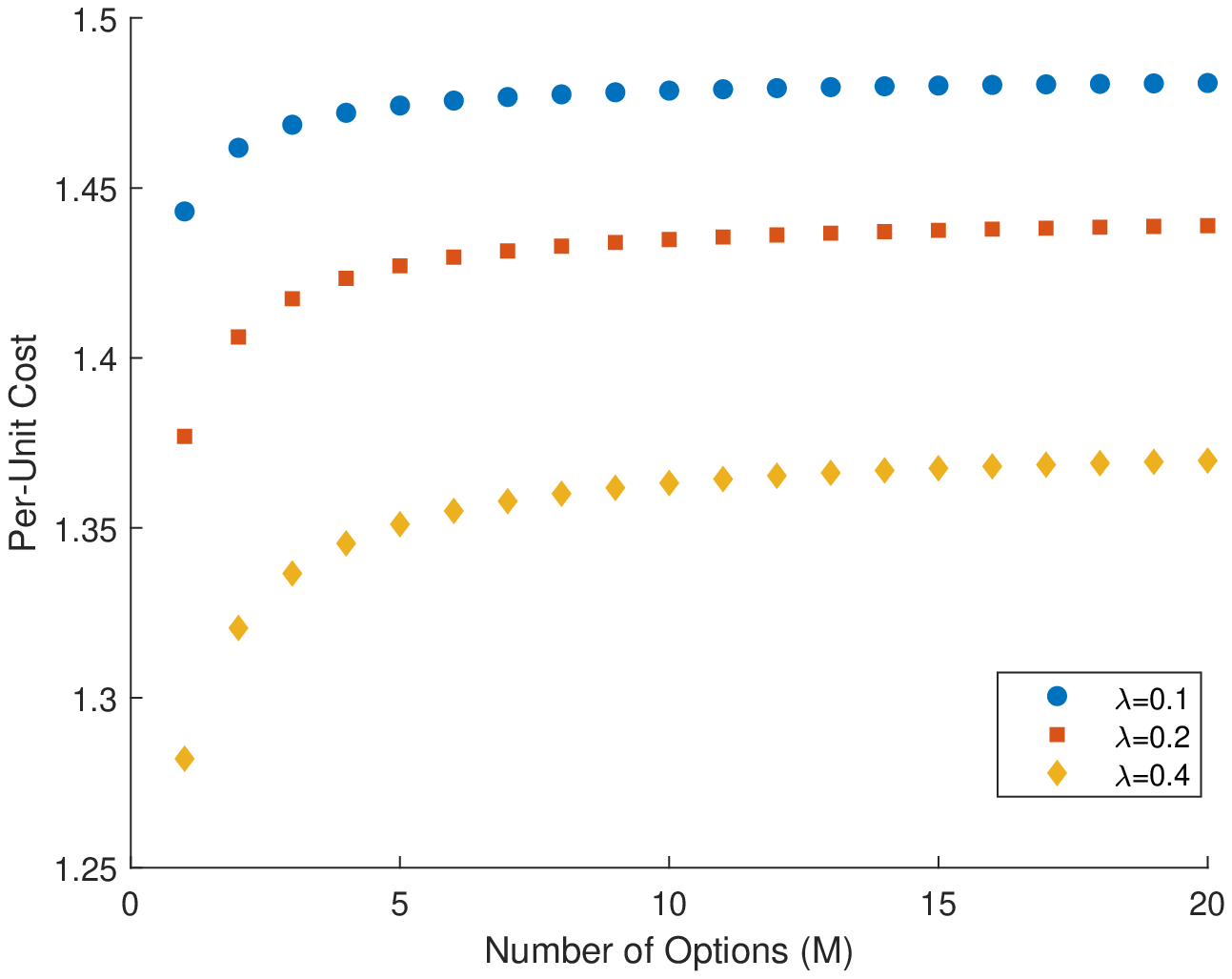}
  \includegraphics[width=2.9in]{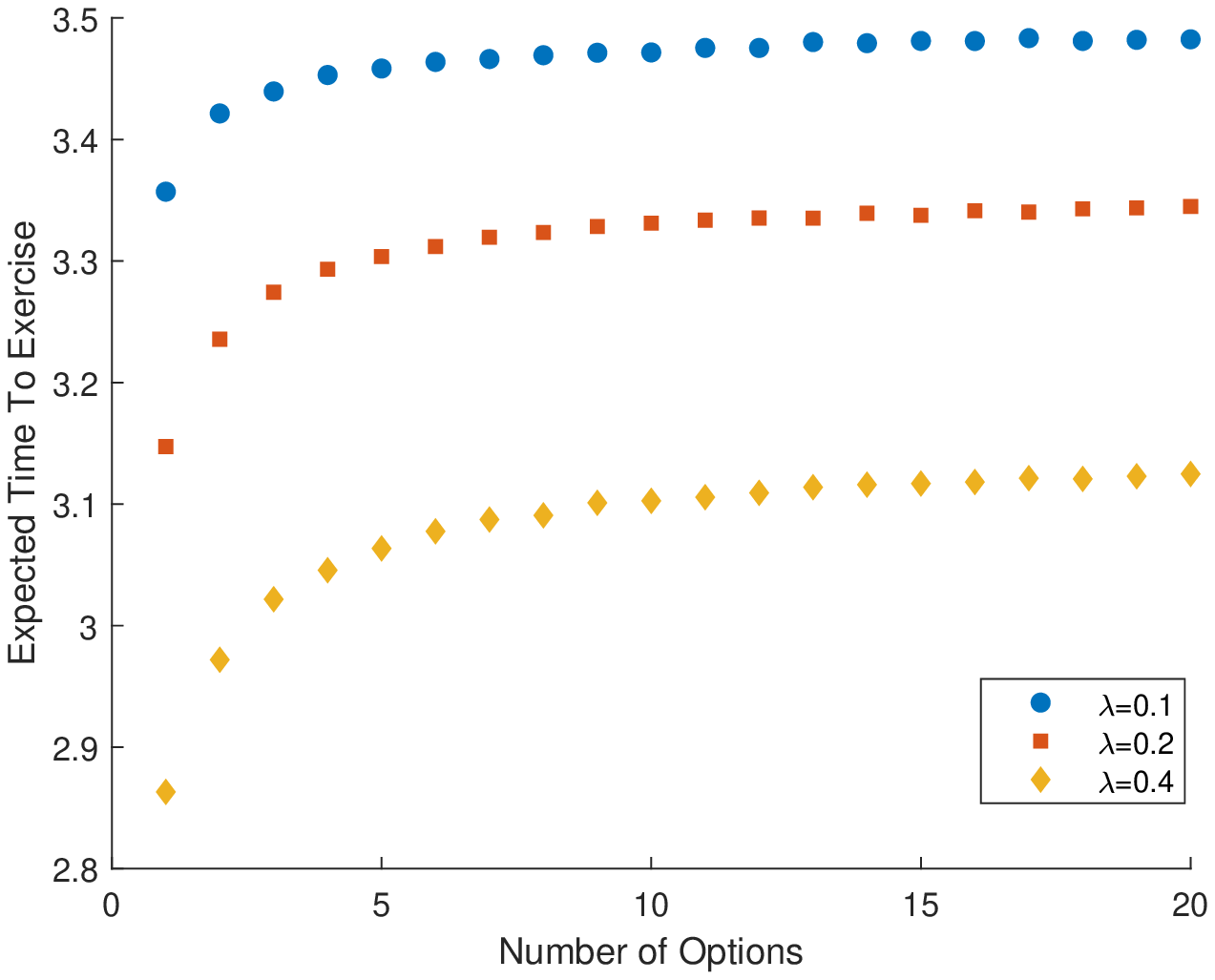}    \caption{  Left:  Per-unit ESO cost as a function of number of options granted  $M$ with different exercise intensities  $\lambda$. Right:  With $M=20$ options under different exercise intensities  $\lambda$, we calculate the average exercise times by Monte-Carlo simulation with $10^8$ simulated paths of exercise process.  Common parameters: $S_0=K=10$, $r=5\%$, $q=1.5\%$, $\sigma=20\%$, $p_{m,z}=1/m$, $T=10$, $t_v=1$, $\beta=0.5$ and
  $\alpha=0.1$. In FFT: $N_x=2^{12}$, $x_{min}=-10$, $x_{max}=10$.}\label{M-AverageExercise}
\end{figure}
\clearpage

%\begin{figure}[!hbpt]
%  \centering
%  % Requires \usepackage{graphicx}
%  \includegraphics[width=4in]{lambda-ESO2.eps}\\
%  \caption{\emph{\textbf{Option Price and Exercise Intensity II.}} We plot the ESO cost as a function of stock price $S_0$ with different exercise intensity ($\lambda$). Parameters: $K=10$, $r=5\%$, $q=1.5\%$, $\sigma=20\%$, $p_{m,z}=1/m$, $M=5$, $T=10$, $t_v=0$ and $\beta=0.1$. In FFT: $N_x=2^{12}$, $x_{min}=-10$, $x_{max}=10$.}\label{lambda-ESO2}
%\end{figure}
% 

%%We'll show the evidence in Table \ref{ComparisonIII}.

%\begin{figure}[!hbpt]
%  \centering
%  % Requires \usepackage{graphicx}
%  \includegraphics[width=4in]{M-ESO.eps}\\
%  \caption{\emph{\textbf{Sublinearity.}} We plot the ESO per-unit cost as a function of number of options ($M$) with different exercise intensity ($\lambda$). Parameters: $S_0=K=10$, $r=5\%$, $q=1.5\%$, $\sigma=20\%$, $p_{m,z}=1/m$, $T=10$, $t_v=1$, $\beta=0.5$ and
%  $\alpha=0.1$. In FFT: $N_x=2^{12}$, $x_{min}=-10$, $x_{max}=10$.}\label{M-ESO}
%\end{figure}
%\clearpage

\newpage

\section{Stochastic Exercise Intensity}\label{sect-stochastic}
Now we  discuss the stochastic exercise intensity, an extension to the previous model, that $\lambda_t=\lambda(t,S_t)$, which is the function not only depends on the time $t$ also depends on the stock price $S_t$. Accordingly, the corresponding vested ESO cost $C^{(m)}(t,s)$ will satisfy
\begin{equation}\label{StochasticExerciseIntensity}
\begin{aligned}
-(r+\lambda(t,s)+\beta)C^{(m)}+C^{(m)}_t+\mathcal{L}C^{(m)}+ \lambda(t,s)\sum_{z=1}^{m-1} p_{m,z}C^{(m-z)}\quad&\\
 +  \left(\lambda(t,s)\bar p_m+m\beta\right) (s-K)^+=0,&
\end{aligned}
\end{equation}
for $m=1,\ldots,M$, and $(t,s)\in[t_v,T]\times\R_+$, with terminal condition $C^{(m)}(T,s) =  m(s-K)^+$, for $s\in\R_+$.

Since we only discuss the stochastic exercise intensity and the employee will not exercise the option during the vesting period, the PDE for unvested ESO will remain unchanged. Next, we will discuss how to numerically solve (\ref{StochasticExerciseIntensity}) by FFT.

%\subsection{Finite Difference Method}
%
%For the finite difference method, the critical issue is the boundary condition. In the previous chapters, we set the boundary condition $s=S_*$ as $C^{(m)}(t,S_*)=\bar C^{(m)}(t,S_*)$ and we can solve $\bar C^{(m)}(t,S_*)$ analytically. However, once the exercise intensity turns to be stochastic, the PDE that $\bar C^{(m)}(t,S_*)$
%satisfies turns to be
%\begin{equation}\label{StochasticBoundaryCondition}
%\begin{aligned}
%-(r+\lambda(t,s)+\beta)\bar C^{(m)}+\bar C^{(m)}_t+\mathcal{L}\bar C^{(m)}+ \lambda(t,s)\sum_{z=1}^{m} p_{m,z}\bar C^{(m-z)}&& \\
% +  \left(\lambda(t,s)\bar p_m+m\beta\right) (s-K)=0&&\ \mbox{for }(t,s)\in(t_v,T)\times\R_+,\\
%\bar C^{(m)}(T,s) =  m(s-K) &&\ \mbox{for }s\in\R_+.
%\end{aligned}
%\end{equation}
%
%We could adopt the same method in Section \ref{FDM} to derive the termination condition. Interested reader could pursue that. However for the sake of convenience, we use Neumann boundary condition here, which is
%\[
%\frac{\partial}{\partial s}C^{(m)}(t,S_*)=m\Rightarrow C^{(m)}(t,S_*+ds)=C^{(m)}(t,S_*)+mdS
%\]
%for $m=1,\ldots,M$.

% 
For  applying Fourier transform, we use the same notation as  in Section \ref{FFT} that
\begin{equation}
f^{(m)}(t,x)=C^{(m)}(t,Ke^x),
\end{equation}
for $m=1,\ldots,M$, $(t,x)\in[t_v,T]\times\R$, and 
\begin{equation}
\mathcal{F}[f^{(m)}](t,\omega)=\int_{-\infty}^{\infty}f^{(m)}(t,x)e^{-i\omega x}dx,\label{FourierTransformS}
\end{equation}
for $m=1,\ldots,M$. In this section, we assume that $\lambda(t,x)=A(t)-B(t)x$, for some positive time-dependent functions $A(t)$ and $B(t)$. For implementation, we assume $B(t)$ be relative small, such that $\lambda(t,x)$ stay positive in the truncated space $(t_v,T)\times[-x_{max}, x_{max}]$. Then, $f^{(m)}(t,x)$ satisfies
\begin{equation}
\begin{aligned}
-(r+A(t)-B(t) x+\beta)f^{(m)}+f^{(m)}_t+\widetilde{\mathcal{L}}f^{(m)}+ (A(t)-B(t) x)\sum_{z=1}^{m-1} p_{m,z}f^{(m-z)}
&\\
  +\bigg(\big(A(t)-B(t)x\big)\bar p_m+m\beta\bigg) (Ke^x-K)^+=0,&\label{FourierPDE1}
\end{aligned}
\end{equation}
where $\widetilde{\mathcal{L}}$ is defined  in (\ref{Lf}).  The terminal condition is  $f^{(m)}(T,x) =  m(Ke^x-K)^+$, for $x\in[-x_{max}, x_{max}]$.

Using   (\ref{FourierTransformS}) and the property of Fourier transform that $\mathcal{F}[xf](t,\omega)=i\partial_\omega\mathcal{F}[f](t,\omega)$, we transform PDE \eqref{FourierPDE1} into  
\begin{align}
\frac{\partial}{\partial t}\mathcal{F}[f^{(m)}](t,\omega)+iB(t)\frac{\partial}{\partial \omega}\mathcal{F}[f^{(m)}](t,\omega)-h(t,\omega)\mathcal{F}[f^{(m)}](t,\omega)+\psi^{(m)}(t,\omega)=0,\label{FourierPDE2}
\end{align}
\newpage
where
\begin{align}
h(t,\omega)&=-(r-q-\frac{\sigma^2}{2})i\omega+\frac{\sigma^2\omega^2}{2}+r+A(t)+\beta,\\
\psi^{(m)}(t,\omega)&=\sum_{z=1}^{m-1} p_{m,z}\mathcal{F}[\lambda f^{(m-z)}](t,\omega)  +  \mathcal{F}[\left(\lambda\bar p_m+m\beta\right) (Ke^x-K)^+](t,\omega),
\end{align}
for $(t,\omega)\in[t_v,T)\times\mathbb{R}$.

Observe that  (\ref{FourierPDE2}) is a first-order PDE with terminal condition that $\mathcal{F}[f^{(m)}](T,\omega) =  m \varphi(\omega)$  (see \eqref{varphi}). Therefore, we apply the method of characteristics and get 
\begin{align}
\mathcal{F}&[f^{(m)}](t,\omega)\notag\\
&=e^{-\int_t^Th\big(s,\omega-i\int_s^{t}B(u)du\big)ds}\mathcal{F}[f^{(m)}]\bigg(T,\omega+i\int_t^{T}B(u)du\bigg)+\int_t^T g^{(m)}(\tau,\omega;t)d\tau.\label{Fbg}
\end{align}
where \begin{equation}\label{gt}
g^{(m)}(\tau,\omega;t)=e^{-\int_t^\tau h\big(s,\omega-i\int_s^{t}B(u)du\big)ds}\psi^{(m)}\bigg(\tau,\omega+i\int^\tau_{t}B(u)du\bigg).
\end{equation} 
%for $(t,\omega)\in[t_v,T)\times\mathbb{R}$.

%where
%\begin{align}
%F^{(m)}(t,\omega)&=\mathcal{F}[f^{(m)}]\bigg(t,\omega+i\int_t^{T}B(u)du\bigg)\\
%G(\tau,\omega)&=e^{-\int_t^\tau h\big(s,\omega+i\int_s^{T}B(u)du\big)ds}\\
%\Phi^{(m)}(t,\omega)&=\psi^{(m)}\bigg(t,\omega+i\int_t^{T}B(u)du\bigg)\\
%&=\mathcal{F}\bigg[e^{x\int_t^{T}B(u)du}\bigg(\sum_{k=0}^{m-1} p_{m,m-k}\lambda f^{(k)} +  \left(\lambda\bar p_m+m\beta\right) (e^x-K)^+\bigg)\bigg](t,\omega)\\
%t_v\le t\le T,&\quad-\infty<x<\infty.\notag
%\end{align}
%It results in
%\begin{align}
%C(t,e^x)=f^{(m)}(t,x)=&\mathcal{F}^{-1}[e^{-\int_t^Th\big(s,\omega+i\int_s^{t}B(u)du\big)ds}\mathcal{F}[e^{-x\int_t^{T}B(u)du}f^{(m)}](T,\omega)\notag\\
%&\quad\quad-\int_t^T e^{-\int_t^\tau h\big(s,\omega+i\int_s^{t}B(u)du\big)ds}\tilde\psi^{(m)}(t,\tau,\omega)d\tau]
%\end{align}
%for $t_v\le t\le T,\quad-\infty<x<\infty$.
%
For numerical implementation, we can use the similar method mentioned in Section \ref{FFT}. We can make the approximation
\begin{align}
\int_t^T g^{(m)}(\tau,\omega;t)d\tau\approx\bigg(\frac{1}{2}g^{(m)}(t,\omega;t)+\frac{1}{2}g^{(m)}(T,\omega;t)+\sum_{i=1}^{i=N-1}g^{(m)}(t+i\delta t,\omega;t)\bigg)\delta t,\label{Approx-g}
\end{align}
where $\delta t=(T-t)/N$. The integral $\int B(u)du$ in \eqref{Fbg} and \eqref{gt} can be  approximated similarly or computed explicitly depending on the choice of $B(t)$.

Table \ref{ComparisonII} presents the ESO costs in the cases of constant exercise intensity $\lambda = 0.2$ and stochastic intensity with   $\lambda(s)=0.2-0.02\log(s/K)$ under different vesting periods and  job termination rates. The stochastic intensity specified here can be larger or smaller than the constant level 0.2 depending on whether the current stock price $s$ is higher or lower than the strike price $K$. For each case, we compute the ESO cost using both FFT and FDM. For the latter, we apply the Crank-Nicolson method on a uniform grid and adopt Neumann condition at the boundary $s=S_*$ (see Section \ref{FDM}).  As we can see, the costs from the two methods are practically the same.  As the vesting period lengthens, from $t_v=1$ to $t_v=4$, the ESO cost tends to increase  under different exercise intensities and job termination rates.  When  the job termination rate $\beta$ is zero, the ESO costs with stochastic intensity appear to be higher, but this effect is greatly reduced as the job termination rate increases. This is intuitive since a high job termination rate means that most ESOs will be exercised or forfeited at the departure time, rather than exercised according to an exercise process over the life of the options.

\begin{table}[!htbp]
  \centering
  \begin{tabular}{ll|cc|cc|cc}
     \hline
     \multicolumn{2}{c}{\multirow{2}{*}{Parameters}} &\multicolumn{2}{|c|}{$t_v=1$}&\multicolumn{2}{|c}{$t_v=2$}&\multicolumn{2}{|c}{$t_v=4$}\\
     \cline{3-8}
      & & FDM & FFT & FDM & FFT & FDM & FFT \\
      \hline\hline
  $\lambda=0.2$ &   \\
            \hline
      \multirow{2}{*}{$\beta=0$} & $\alpha=0$ & 12.8052  & 12.8065&   13.7122 &  13.7134 &  15.0953  & 15.0967\\
      & $\alpha=0.1$& 11.5867   &11.5878 &  11.2266 &  11.2276   & 10.1187  & 10.1196\\
     \hline
      \multirow{2}{*}{$\beta=0.5$} & $\alpha=0$ &7.8849  &  7.8859&    9.6380   & 9.6388  & 12.4022 &  12.4029\\
      & $\alpha=0.1$&7.1347  &  7.1355  &  7.8910 &   7.8916 &   8.3135  &  8.3139\\
     \hline\hline
%      \hline 
%      &\\
           $\lambda(s)=0.2-0.02$\!\!\!\!\!&$*\log({s}/{K})$ &    \\
      \hline
      \multirow{2}{*}{$\beta=0$} & $\alpha=0$   & 12.8310 &  12.8379   &13.7364  & 13.7445&15.1130&15.1235\\
      & $\alpha=0.1$&11.6099  & 11.6163 &  11.2464 &  11.2531 &10.1305 &  10.1376\\
     \hline
      \multirow{2}{*}{$\beta=0.5$} & $\alpha=0$ & 7.8895   & 7.8887&    9.6428  &  9.6423&12.4068 &  12.4076\\
      & $\alpha=0.1$&   7.1387  &  7.1381  &  7.8948  &  7.8946 &8.3165  &  8.3172\\
     \hline

   \end{tabular}
  \caption{ESO costs with constant intensity $\lambda$ and stochastic exercise intensity $\lambda(s)$ with different job termination rates $\alpha$ and $\beta$ and vesting period $t_v$,  computed using FFT and FDM for comparison. Common parameters: $S_0=K=10$, $r=5\%$, $q=1.5\%$, $\sigma=20\%$, $p_{m,z}=1/m$, $M=5$, $T=10$. In FDM: $S_*=30$, $\delta S=0.1$, $\delta t=0.1$. In FFT: $N_x=2^{12}$, $X_{min}=-10$ and $X_{max}=10$.}\label{ComparisonII}
\end{table}

 \clearpage

%\begin{table}[h]
%  \centering
%  \begin{tabular}{cc|cc|cc|cc}
%     \hline
%     \multicolumn{2}{c}{\multirow{2}{*}{Parameters}} &\multicolumn{2}{|c}{$t_v=1$}&\multicolumn{2}{|c}{$t_v=2$}&\multicolumn{2}{|c}{$t_v=4$}\\
%     \cline{3-8}
%     % after \\: \hline or \cline{col1-col2} \cline{col3-col4} ...
%     & & FDM & FFT & FDM&FFT &FDM &FFT\\
%      \hline
%      \hline
%      \multirow{2}{*}{$\beta=0$} & $\alpha=0$   & 12.8310 &  12.8379   &13.7364  & 13.7445&15.1130&15.1235\\
%      & $\alpha=0.1$&11.6099  & 11.6163 &  11.2464 &  11.2531 &10.1305 &  10.1376\\
%     \hline
%      \multirow{2}{*}{$\beta=0.5$} & $\alpha=0$ & 7.8895   & 7.8887&    9.6428  &  9.6423&12.4068 &  12.4076\\
%      & $\alpha=0.1$&   7.1387  &  7.1381  &  7.8948  &  7.8946 &8.3165  &  8.3172\\
%     \hline
%   \end{tabular}
%  \caption{\emph{\textbf{ESO cost comparison under stochastic exercise intensity.}} Common Parameters: $S_0=K=10$, $r=5\%$, $q=1.5\%$, $\sigma=20\%$, $p_{m,z}=1/m$, $M=5$, $T=10$ and
%  $\lambda(s)=0.2-0.02\log(s/K)$. In FDM: $S_*=50$, $\delta S=0.1$, $\delta t=0.1$. In FFT: $\delta_t=0.01$, $N_x=2^{12}$, $X_{min}=-10$ and $X_{max}=10$.}\label{ComparisonV}
%\end{table}

 %%%%%%%%%%%%%%%%%%%%%
\section{Maturity Randomization}\label{sect-mat-rand}
In this section,  we propose an alternative method to value ESOs. It is an analytical method that yields an approximation to the original ESO valuation problem discussed in Section \ref{sect-ESOPDE}. The core idea of this method is to randomize the ESO's finite maturity by an exponential random variable $\tau\sim\exp(\kappa)$, with $\kappa=1/T$ where $T$ here is original constant maturity. Such a choice of parameter means that  $\E[\tau]=T$; that is, the ESO is expected to expire at time $T$. For instance, if the maturity of the ESOs is 10 years, then the randomized maturity is modeled by $\tau\sim\exp(0.1)$. Such a maturity randomization allows us to derive an explicit approximation for ESO cost. 
%Some researchers assume the ESO be perpetual in order to get the explicit solution, like \cite{Carmona11}, which is equivalent to the case $\kappa=0$.

\subsection{Vested ESO}
First we consider  the ESO cost at the end of the vesting period. Provided that the employee remains at the firm by time $t_v$,   the vested ESO has a remaining maturity of length $T-t_v$. Therefore, for the exponentially distributed maturity $\tau\sim\exp(\kappa)$, one may set $\kappa=1/(T-t_v)$.  At time $t_v$, the vested ESO cost function $C^{(m)}(s)$ is given by
\begin{align}
C^{(m)}(s)=& \E\bigg\{ \int_{t_v}^{\tau\wedge \xi} e^{-r (u-t_v)} (S_u -K)^+ dL_u \notag \\
&+ e^{-r(\tau\wedge \xi-t_v)} (M- L_{\tau\wedge \xi}) (S_{\tau\wedge \xi}-K)^+\,|\, S_{t_v} = s, L_{t_v} = M-m, \tau\wedge\xi\ge {t_v}\bigg\}\\
=& \E\bigg\{ \int_{t_v}^\infty e^{-(r+\kappa+\beta) (u-t_v)} (S_u -K)^+ dL_u \notag \\
&+ \int_{t_v}^\infty (\kappa+\beta)e^{-(r+\kappa+\beta)(u-t_v)} (M- L_u) (S_u-K)^+du\,|\, S_{t_v} = s, L_{t_v} = M-m\bigg\}\\
=& \E\bigg\{ \int_{0}^\infty e^{-(r+\kappa+\beta) u} (S_u -K)^+ dL_u \notag \\
&+ \int_0^\infty (\kappa+\beta)e^{-(r+\kappa+\beta)u} (M- L_u) (S_u-K)^+du\,|\, S_0 = s, L_0 = M-m \bigg\},\label{MR}
\end{align}
for $m=1,\ldots,M$.  From  (\ref{MR}), we derive the associated  ODE for $C^{(m)}(s)$. For the convenience, we denote
\begin{equation}
\begin{aligned}
a_0=-(r+\lambda+\beta+\kappa),\quad a_1=r-q,\quad a_2=\frac{\sigma^2}{2},\quad
g_m=  \lambda\bar p_m+m (\beta+ \kappa ).
\end{aligned}
\end{equation}
Then,  we obtain a system of second-order linear ODEs: 
\begin{equation}\label{RMT_PDE}
a_0C^{(m)}+a_1s\frac{d}{ds}C^{(m)}+a_2s^2\frac{d^2}{ds^2}C^{(m)}+\lambda\sum_{z=1}^{m-1}p_{m,z} C^{(m-z)}  +g_m(s-K)^+=0,
\end{equation}
for $m=1,\ldots,M$, and $s\in\R_+$, with the boundary condition $C^{(m)}(0) =  0$.

%\begin{equation}\label{RMT_PDE}
%-(r+\lambda+\beta+\kappa)C^{(m)}+\mathcal{L}C^{(m)}+\lambda\sum_{z=1}^{m}p_{m,z} C^{(m-z)}  +  \left(\lambda\bar p_m+m\beta+m\kappa\right) (s-K)^+=0,
%\end{equation}
%for $m=1,\ldots,M$, and $s\in\R_+$, where the operator $\mathcal{L}$ is defined by equation (\ref{L}), with the boundary condition $C^{(m)}(0) =  0$.

 \newpage
\begin{proposition}\label{prop-RM}
The solution to the ODE system (\ref{RMT_PDE}) is
\begin{equation}
C^{(m)}(s)=
\left\{
\begin{aligned}
&A_ms+B_mK+\sum_{n=0}^{m-1}E_{m,n}[\ln(\frac{s}{K})]^n(\frac{s}{K})^{\gamma-\theta};&\quad \mbox{if }s>K,\\
&\sum_{n=0}^{m-1}F_{m,n}[\ln(\frac{s}{K})]^n(\frac{s}{K})^{\gamma+\theta};&\quad \mbox{if }0\le s\le K,\label{RM_vestedESOsolution}
\end{aligned}
\right.
\end{equation}
for $m=1,\ldots,M$, where
 \begin{equation}
\left\{
\begin{aligned}
A_m&=\frac{1}{a_1+a_0}\left(-\lambda\sum_{z=1}^{m-1}p_{m,z}A_{m-z}-g_m\right), &\\
B_m&=\frac{1}{a_0}\left(-\lambda\sum_{z=1}^{m-1}p_{m,z}B_{m-z}+g_m\right),  &\\
E_{1,0}&=-\frac{(A_1+B_1)K(\gamma+\theta)-A_1K}{2\theta},\\
F_{1,0}&=-\frac{(A_1+B_1)K(\gamma-\theta)-A_1K}{2\theta},\\
E_{m,m-1}&=-\frac{\lambda p_{m,1}E_{m-1,m-2}}{(m-1)[a_1+2a_2(\gamma-\theta)-a_2]}, &\quad \mbox{for }m\ge 2,\\
F_{m,m-1}&=-\frac{\lambda p_{m,1}F_{m-1,m-2}}{(m-1)[a_1+2a_2(\gamma+\theta)-a_2]},&\quad \mbox{for }m\ge 2,\\
E_{m,n}&=-\frac{\lambda\sum_{z=1}^{m-n}p_{m,z}E_{m-z,n-1}+(n+1)na_2E_{m,n+1}}{n[a_1+2a_2(\gamma-\theta)-a_2]}, &\quad \mbox{for }1\le n\le m-2,\\
F_{m,n}&=-\frac{\lambda\sum_{z=1}^{m-n}p_{m,z}F_{m-z,n-1}+(n+1)na_2F_{m,n+1}}{n[a_1+2a_2(\gamma+\theta)-a_2]}, &\quad \mbox{for }1\le n\le m-2,\\
E_{m,0}&=-\frac{(A_m+B_m)K(\gamma+\theta)-A_mK+F_{m,1}-E_{m,1}}{2\theta},&\quad \mbox{for }m\ge 2,\\
F_{m,0}&=-\frac{(A_m+B_m)K(\gamma-\theta)-A_mK+F_{m,1}-E_{m,1}}{2\theta},&\quad \mbox{for }m\ge 2,\\
\end{aligned}
\right.
\end{equation}
and
\begin{align}
\gamma=\frac{1}{2}-\frac{r-q}{\sigma^2},\quad \theta=\sqrt{\gamma^2+\frac{2(r+\lambda+\beta+\kappa)}{\sigma^2}}.
\end{align}
\end{proposition}

%\begin{proof}
%For the convenience, we denote
%\begin{equation}
%\begin{aligned}
%a_0=-(r+\lambda+\beta+\kappa),\quad a_1=r-q,\quad a_2=\frac{\sigma^2}{2},\quad
%g_m=  \left(\lambda\bar p_m+m\beta+m\kappa\right).
%\end{aligned}
%\end{equation}
%Then, the ODE (\ref{RMT_PDE}) can  be written as
%\begin{equation}
%a_0C^{(m)}+a_1s\frac{d}{ds}C^{(m)}+a_2s^2\frac{d^2}{ds^2}C^{(m)}+\lambda\sum_{z=1}^{m}p_{m,z} C^{(m-z)}  +g_m(s-K)^+=0,
%\end{equation}
%for $m=1,\ldots,M$, and $s\in\R_+$.
%
%Therefore, we have the ansatz solution for ODE (\ref{RMT_PDE})
%\begin{equation}
%C^{(m)}(s)=
%\left\{
%\begin{aligned}
%&A_ms+B_mK+&\\
%&\sum_{n=0}^{m-1}\bigg(E_{m,n}[\ln(\frac{s}{K})]^n(\frac{s}{K})^{\gamma-\theta}+\tilde E_{m,n}[\ln(\frac{s}{K})]^n(\frac{s}{K})^{\gamma+\theta}\bigg);&\quad \mbox{if }s>K,\\
%&\sum_{n=0}^{m-1}\bigg(F_{m,n}[\ln(\frac{s}{K})]^n(\frac{s}{K})^{\gamma+\theta}+\tilde F_{m,n}[\ln(\frac{s}{K})]^n(\frac{s}{K})^{\gamma-\theta}\bigg);&\quad \mbox{if }0\le s\le K,\label{RM_vestedESOsolution}
%\end{aligned}
%\right.
%\end{equation}
%where
%\begin{equation}
%\gamma=\frac{a_2-a_1}{2a_2},\quad \theta=\sqrt{\gamma^2-\frac{a_0}{a_2}}.
%\end{equation}
%\end{proof}

%Here, $p(m,z)$ represents the probability of selling $z$ options when the employee holds $m$ options and $\bar p_m$ represents the expected number of options that to be sold. Hence, we could derive the explicit solution as followings.
\begin{proof}
We begin by considering the case that the employee only holds a single option. With $M=1$, the  general solution to ODE \eqref{RMT_PDE} is given by
\begin{equation}\label{c1s}
C^{(1)}(s)=
\left\{
\begin{aligned}
&A_1s+B_1K+E_{1,0}(\frac{s}{K})^{\gamma-\theta}+\tilde {E}_{1,0}(\frac{s}{K})^{\gamma+\theta};&\quad \mbox{if }s>K,\\
&F_{1,0}(\frac{s}{K})^{\gamma+\theta}+\tilde {F}_{1,0}(\frac{s}{K})^{\gamma-\theta};&\quad \mbox{if }0\le s\le K,
\end{aligned}
\right.
\end{equation}
where 
\begin{equation}
A_1=-\frac{g_1}{a_1+a_0}, \quad B_1=\frac{g_1}{a_0}.  
\end{equation}
%where
%\begin{equation}
%\gamma=\frac{1}{2}-\frac{r-q}{\sigma^2},\quad \theta=\sqrt{\gamma^2+\frac{2(r+\lambda+\beta+\kappa)}{\sigma^2}}, \quad g=\lambda+\beta+\kappa.
%\end{equation}
By imposing that $C^{(1)}(s)$ and $\frac{d}{ds}C^{(1)}(s)$ to be  continuous at the strike price $K$, we consider $C^{(1)}(s)$ at $s=K$ and obtain 
\begin{align}
&\bcm 1&1\\ \gamma-\theta& \gamma+\theta \ecm \bcm E_{1,0}-\tilde F_{1,0}\\\tilde E_{1,0}-F_{1,0}\ecm=-K\bcm A_1+B_1\\A_1\ecm \\
\Rightarrow & \bcm E_{1,0}-\tilde F_{1,0}\\\tilde E_{1,0}-F_{1,0}\ecm=-\frac{K}{2\theta}\bcm (\gamma+\theta)(A_1+B_1)-A_1\\ -(\gamma-\theta)(A_1+B_1)+A_1\ecm.
\end{align}
%In addition,
%\begin{align*}
%\gamma+\theta&>\gamma+\sqrt{\gamma^2+\frac{2r}{\sigma^2}}\\
%&=\frac{1}{2}-\frac{r-q}{\sigma^2}+\sqrt{(\frac{1}{2}-\frac{r-q}{\sigma^2})^2+\frac{2r}{\sigma^2}}\\
%&>\frac{1}{2}-\frac{r}{\sigma^2}+\sqrt{(\frac{1}{2}-\frac{r}{\sigma^2})^2+\frac{2r}{\sigma^2}}\\
%&=1
%\end{align*}
%\begin{equation}
%\gamma-\theta<0
%\end{equation}
 In addition, since $\gamma-\theta<0$, we will have $\tilde{F}_{1,0}=0$ to guarantee that $C^{(1)}(0)=0$. And, when $\kappa\rightarrow \infty$, the maturity $\tau\rightarrow 0 $, $\mathbb{P}$-a.s., which will lead to $C^{(1)}(s)\rightarrow(s-K)^+$. Therefore, we have $\tilde {E}_{1,0}=0$. As a result, we obtain the remaining non-zero coefficients:
\begin{equation}
\bcm E_{1,0}\\F_{1,0}\ecm=-\frac{K}{2\theta}\bcm (\gamma+\theta)(A_1+B_1)-A_1\\ (\gamma-\theta)(A_1+B_1)-A_1\ecm.
\end{equation}
For $M\ge 2$, the general solution to  ODE (\ref{RMT_PDE}) is
\begin{equation}
C^{(m)}(s)=
\left\{
\begin{aligned}
&A_ms+B_mK+\sum_{n=0}^{m-1}E_{m,n}[\ln(\frac{s}{K})]^n(\frac{s}{K})^{\gamma-\theta}&\quad \mbox{if }s>K,\\
&\sum_{n=0}^{m-1}F_{m,n}[\ln(\frac{s}{K})]^n(\frac{s}{K})^{\gamma+\theta}&\quad \mbox{if }0\le s\le K.\label{RM_vestedESOsolution}
\end{aligned}
\right.
\end{equation}
Applying ODE (\ref{RMT_PDE}), we obtain the relationship between the coefficients of $C^{(m)}(s)$ and the coefficients of $C^{(n)}(s)$, for $n\le m-1$, as follows:
 \begin{equation}
\left\{
\begin{aligned}
A_m&=\frac{1}{a_1+a_0}\left(-\lambda\sum_{z=1}^{m-1}p_{m,z}A_{m-z}-g_m\right), &\\
B_m&=\frac{1}{a_0}\left(-\lambda\sum_{z=1}^{m-1}p_{m,z}B_{m-z}+g_m\right),&\\
E_{m,m-1}&=-\frac{\lambda p_{m,1}E_{m-1,m-2}}{(m-1)[a_1+2a_2(\gamma-\theta)-a_2]}, &\\
F_{m,m-1}&=-\frac{\lambda p_{m,1}F_{m-1,m-2}}{(m-1)[a_1+2a_2(\gamma+\theta)-a_2]},&\\
E_{m,n}&=-\frac{\lambda\sum_{z=1}^{m-n}p_{m,z}E_{m-z,n-1}+(n+1)na_2E_{m,n+1}}{n[a_1+2a_2(\gamma-\theta)-a_2]}, &\quad \mbox{for }1\le n\le m-2,\\
F_{m,n}&=-\frac{\lambda\sum_{z=1}^{m-n}p_{m,z}F_{m-z,n-1}+(n+1)na_2F_{m,n+1}}{n[a_1+2a_2(\gamma+\theta)-a_2]}, &\quad \mbox{for }1\le n\le m-2,\\
\end{aligned}
\right.
\end{equation}
for $m=2,\ldots,M$.

In addition,   the continuity of $C^{(m)}(s)$ and $\frac{d}{ds}C^{(m)}(s)$ around strike price $K$ yields that 
 \begin{equation}
\left\{
\begin{aligned}
(A_m+B_m)K+E_{m,0}=F_{m,0},\\
A_m+(\gamma-\theta)\frac{E_{m,0}}{K}+\frac{E_{m,1}}{K}=(\gamma+\theta)\frac{F_{m,0}}{K}+\frac{F_{m,1}}{K}.
\end{aligned}
\right.
\end{equation}
Rearranging, we obtain the remaining coefficients for the solution:
 \begin{equation}
\left\{
\begin{aligned}
E_{m,0}=-\frac{(A_m+B_m)K(\gamma+\theta)-A_mK+F_{m,1}-E_{m,1}}{2\theta},\\
F_{m,0}=-\frac{(A_m+B_m)K(\gamma-\theta)-A_mK+F_{m,1}-E_{m,1}}{2\theta}.
\end{aligned}
\right.
\end{equation}
\end{proof}
\subsection{Unvested ESO}

For the unvested ESO, we can model the vesting time $t_v$ by the exponential random variable $\tau_v\sim\exp(\tilde \kappa)$, where $\tilde \kappa=1/t_v$. Then, the unvested ESO cost at time $0$ is given by 
\begin{align}
\tilde C^{(m)}(s)&=\E\bigg\{e^{-(r+\alpha)\tau_v}C^{(m)}(S_{\tau_v})\bigg|S_0=s \bigg\}\\
&=\E\bigg\{\int_0^\infty \tilde \kappa e^{-(r+\alpha+\tilde\kappa)u}C^{(m)}(S_{u})du\bigg|S_0=s\bigg\}.
\end{align}
Then we will can derive the ODE for $\tilde C^{(m)}(s)$:
\begin{equation}
\begin{aligned}
-(r+\alpha+\tilde\kappa)\tilde C^{(m)}+(r-q)s\frac{d}{ds}\tilde C^{(m)}+\frac{\sigma^2s^2}{2}\frac{d^2}{ds^2}\tilde C^{(m)}+\tilde\kappa C^{(m)}=0&\quad\mbox{for }s\in\R_+,\\
\tilde C^{(m)}(0)=0.
\end{aligned}
\end{equation}
Assuming $\lambda+\beta+\kappa\neq\alpha+\tilde\kappa$, we could derive the solution for $\tilde C^{(m)}$ from the solution for $C^{(m)}$ in (\ref{RM_vestedESOsolution}), which is
\begin{equation}
\tilde C^{(m)}(s)=
\left\{
\begin{aligned}
&\tilde A_ms+\tilde B_mK+\sum_{n=0}^{m-1}\tilde E_{m,n}[\ln(\frac{s}{K})]^n(\frac{s}{K})^{\gamma-\theta}+\tilde E_m(\frac{s}{K})^{\tilde\gamma-\tilde\theta}&\quad \mbox{if }s>K,\\
&\sum_{n=0}^{m-1}\tilde F_{m,n}[\ln(\frac{s}{K})]^n(\frac{s}{K})^{\gamma+\theta}+\tilde F_m(\frac{s}{K})^{\tilde\gamma+\tilde\theta}&\quad \mbox{if }0\le s\le K,
\end{aligned}
\right.
\end{equation}
where
\begin{equation}
\left\{
\begin{aligned}
\tilde\gamma&=\gamma=\frac{1}{2}-\frac{r-q}{\sigma^2},\\
\tilde\theta&=\sqrt{\tilde\gamma^2+\frac{2(r+\alpha+\tilde\kappa)}{\sigma^2}},\\
\tilde A_m&=\frac{\tilde\kappa A_m}{q+\alpha+\tilde\kappa},\\
\tilde B_m&=\frac{\tilde\kappa B_m}{r+\alpha+\tilde\kappa},
\end{aligned}
\right.
\end{equation}
and
\begin{equation}
\left\{
\begin{aligned}
\tilde E_{m,m-1}&=-\frac{\tilde\kappa E_{m,m-1}}{R},\\
\tilde F_{m,m-1}&=-\frac{\tilde\kappa F_{m,m-1}}{R},\\
\tilde E_{m,m-2}&=-\frac{\tilde\kappa E_{m,m-2}+(m-1)P_1\tilde E_{m,m-1}}{R},\\
\tilde F_{m,m-2}&=-\frac{\tilde\kappa F_{m,m-2}+(m-1)Q_1\tilde F_{m,m-1}}{R},\\
\tilde E_{m,n}&=-\frac{2\tilde\kappa E_{m,n}+2(n+1)P_1\tilde E_{m,n+1}+\sigma^2(n+2)(n+1)\tilde E_{m,n+2}}{2R}\quad\mbox{for }0\le n\le m-3,\\
\tilde F_{m,n}&=-\frac{2\tilde\kappa F_{m,n}+2(n+1)Q_1\tilde F_{m,n+1}+\sigma^2(n+2)(n+1)\tilde F_{m,n+2}}{2R}\quad\mbox{for }0\le n\le m-3,\\
\tilde E_m&=\frac{(\tilde\gamma+\tilde\theta)P-Q}{2\tilde\theta},\\
\tilde F_m&=\frac{(\tilde\gamma-\tilde\theta)P-Q}{2\tilde\theta},
\end{aligned}
\right.
\end{equation}
with
\begin{equation}
\left\{
\begin{aligned}
R&=\lambda+\beta+\kappa-\alpha-\tilde\kappa,\\
P_1&=r-q+\frac{\sigma^2(2\gamma-2\theta-1)}{2},\\
Q_1&=r-q+\frac{\sigma^2(2\gamma+2\theta-1)}{2},\\
P&=\tilde F_{m,0}-\tilde E_{m,0}-K\tilde A_m-K\tilde B_m,\\
Q&=(\gamma+\theta)\tilde F_{m,0}-(\gamma-\theta)\tilde E_{m,0}-K\tilde A_m+\tilde F_{m,1}-\tilde E_{m,1}.
\end{aligned}
\right.
\end{equation}

Alternatively, one can use FDM or FFT to calculate the   unvested ESO cost without applying maturity randomization for the second time. \\

%\begin{figure}[!hbpt]
%\begin{subfigure}[b]{0.45\textwidth}
%\includegraphics[width=\textwidth]{m1Genaral.eps}
%\caption{$m=1$}
%\end{subfigure}
%~
%\begin{subfigure}[b]{0.45\textwidth}
%\includegraphics[width=\textwidth]{m5Genaral.eps}
%\caption{$m=5$}
%\end{subfigure}
%\caption{ESO cost derived by maturity randomization. Parameters: $S_0=K=10$, $r=5\%$, $q=1.5\%$, $\sigma=20\%$, $p_{m,z}=1/m$ and $\beta=0$.}\label{general}
%\end{figure}

%\subsection{Numerical Implementation}

 In Figure  \ref{ESOCost}, we show the cost of an unvested ESO, computed by our maturity randomization method, as a function of the initial stock price $S_0$, along with the ESO payoff.  As expected, the ESO cost is increasing convex in $S_0$. Comparing the costs corresponding to two different job termination rates  $\alpha \in\{0.01,0.1\}$ during the vesting period, we see that a higher job termination rate reduces the ESO value. This is intuitive as the employee has a higher chance of leaving the firm during the vesting period and in turn losing the option entirely.

The maturity randomization  method delivers an analytical approximation that allows for instant computation.   In Figure \ref{Respective}, we examine errors of this method. As we can see, as the exercise intensity $\lambda$ or post vesting job termination rate $\beta$ increases the valuation error decreases exponentially to less than a penny for each option.  This shows that the maturity randomization method can be very accurate and effective for ESO valuation. 

%Intuitively, if the employee tends to sell the last batch of options before they expire, then the exercises are then this error will be relatively small. Thus, the error is decreasing with respect to job termination rate $\beta$ and exercise intensity $\lambda$, which is shown as . 

\begin{figure}[h]
  \centering
  \includegraphics[width=3in]{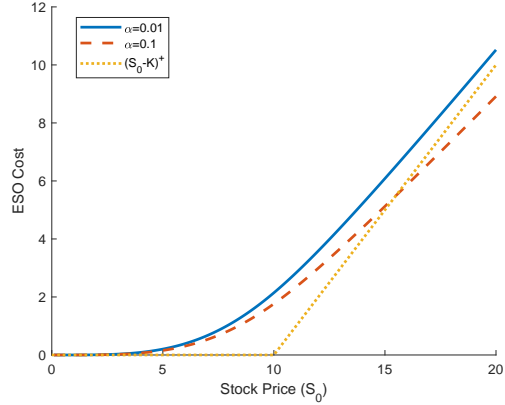}\\
  \caption{The ESO cost computed using the maturity randomization method, and plotted as a function of stock price $S_0$ with two different   job termination rates $\alpha=0.01, 0.1$, along with the ESO payoff function $(S_0-K)^+$ for comparison. Parameters: $T=10$, $t_v=2$, $\kappa=0.125$, $\tilde \kappa=0.5$, $r=5\%$, $q=1.5\%$, $\sigma=20\%$, $\lambda=0.1$ and $\beta=0.1$.}\label{ESOCost}
\end{figure}

\vspace{10pt}

\begin{figure}[h]
\centering
\subfigure[ ]
{\includegraphics[width=3in, angle=0]{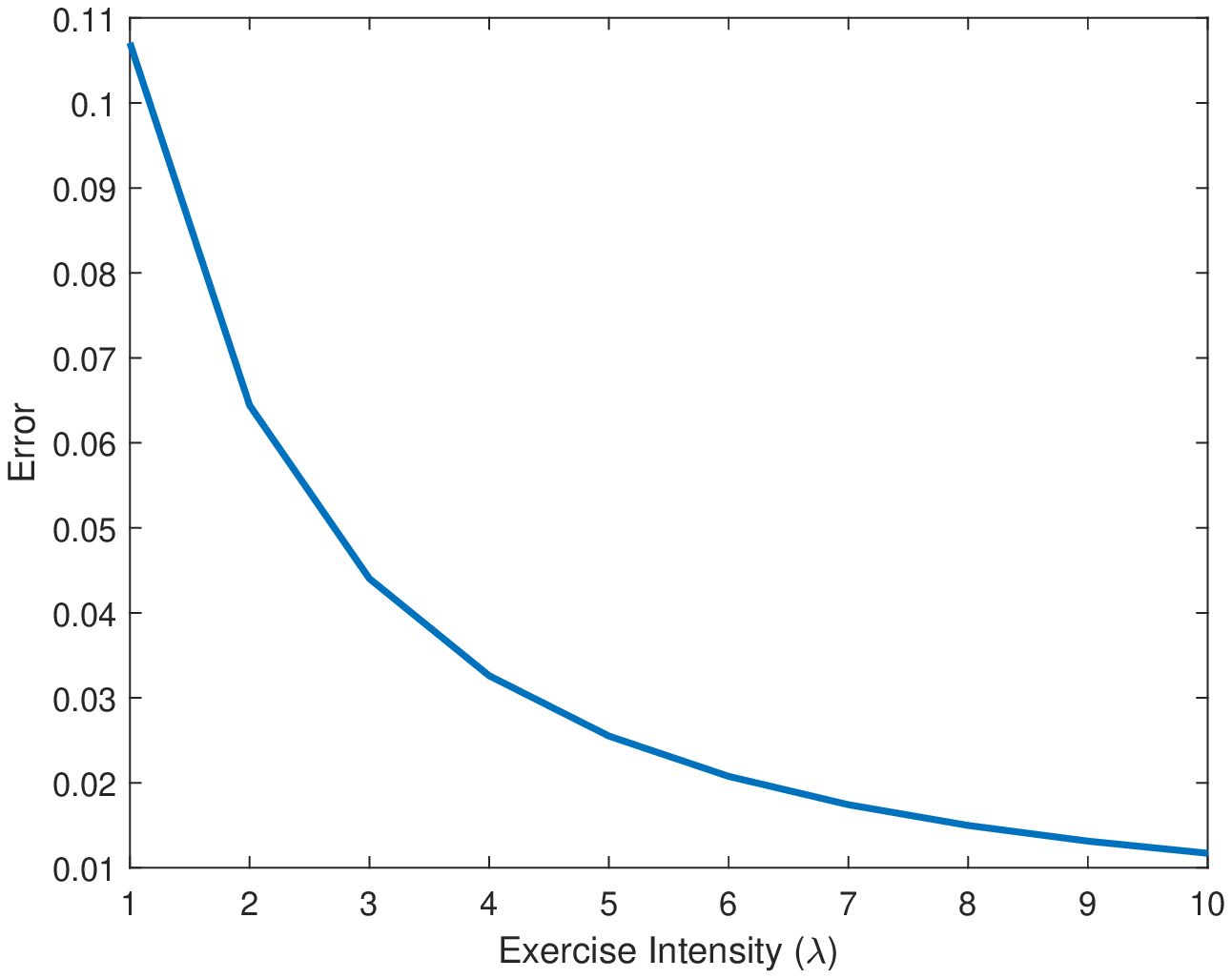}
}
\subfigure[ ]
{\includegraphics[width=3in, angle=0]{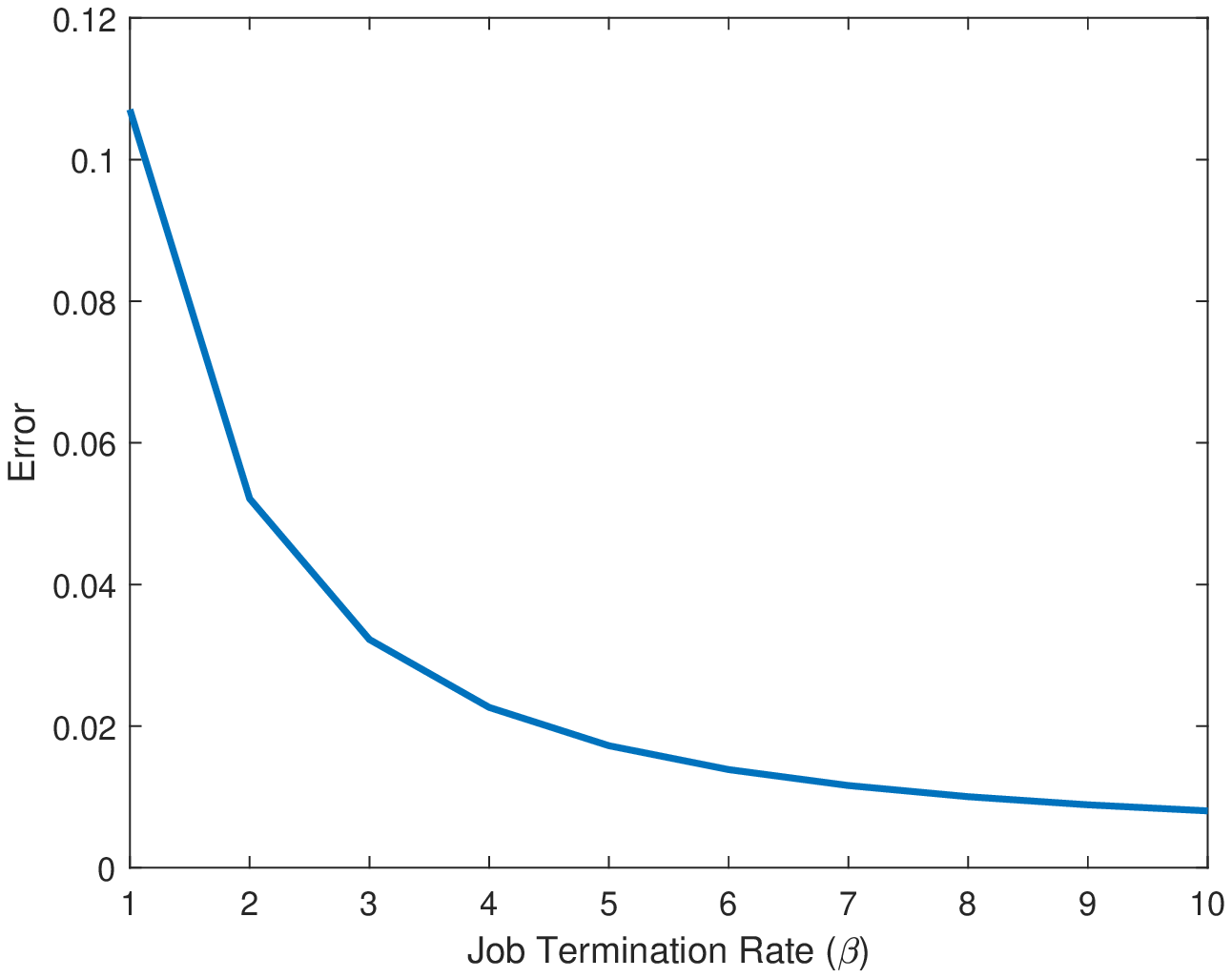}
}
\caption{Plots of the errors of maturity randomization method as a function of exercise intensity $\lambda$ and job termination rate $\beta$ respectively in (a) and (b). We fix $\beta=1$ in (a) and $\lambda=1$ in (b). Parameters: $S_0=K=10$, $T=10$, $t_v=0$, $\kappa=0.1$, $\tilde\kappa=0$, $r=5\%$, $q=1.5\%$, $\sigma=20\%$, $p_{m,z}=1/m$ and $M=5$. Common parameters: $\lambda=1$ and $\beta=1$.}\label{Respective}
\end{figure}
\clearpage

\section{Implied Maturity}
Given that ESOs are very likely to be exercised prior to expiration, the total cost of an ESO grant is determined by how long the employee effectively holds the options. For each grant of $M$ options, the exercise times are different and they depend on the valuation model and associated parameters. Therefore, we introduce the notion of implied maturity to give an intuitive measure of the effective maturity implied by any given valuation model. 

Like the well-known concept of implied volatility, we use the Black-Scholes   option pricing formula. The price of a European call with strike $K$ and maturity $T$ is given by
\begin{equation}
C_{BS}(S_t, {T}) = e^{-q(T-t)}S_t\Phi(d_1)-e^{-r(T-t)}K\Phi(d_2),
\end{equation}
where
\begin{equation}
d_1=\frac{1}{\sigma\sqrt{T-t}}\left[\ln\bigg(\frac{S_t}{K}\bigg)+\bigg(r+\frac{\sigma^2}{2}\bigg)(T-t)\right],\quad d_2=d_1-\sigma\sqrt{T-t}.
\end{equation}
Next, recall the ESO cost function $C^{(m)}(t,s)$ under the top-down valuation model in Section \ref{sect-ESOmodel}. Then,  the \textit{implied maturity} for $m$ ESOs is defined to be the maturity parameter $\tilde{T}$ such that 
\begin{equation}\label{impliedmat}
 C_{BS}(S_0,\widetilde{T}\,) \,= \,\frac{C^{(m)}(0,S_0)}{m}
\end{equation} holds, with all other parameters held constant. To define implied maturity under another model only requires replacing the corresponding cost function on the right-hand side in \eqref{impliedmat}.

Through the lens of implied maturity, we can see the model and parameter effects in terms of how long the employee will hold the option under the Black-Scholes model.  For example, if the exercise intensity $\lambda$ increases, then the ESOs are more likely to be exercised early, resulting in a lower cost. Since the call option value is increasing in maturity, the implied maturity is expected to decrease as exercise intensity increases.  The plots in Figure \ref{lambda-IM} confirm this intuition. Moreover, under high exercise intensity all ESOs will be exercised very early and the contract maturity will play a lesser role on the ESO cost and thus implied maturity. Indeed,  Figure \ref{lambda-IM}  shows that implied maturities associated with different contract  maturities $T=5, 8$ and $10$ get closer to each other as $\lambda$ increases. 

%On the right panel of Figure \ref{lambda-IM}, we plot the implied maturity, computed using maturity randomization.  It shows that, the maturity randomization method underestimates the implied maturity. It is essentially because, in maturity randomization method, we compute the expectation of ESO cost over the distribution of maturity and the ESO cost is convex with respect to the maturity.  

Next, we consider  the effect of the total number of ESOs granted. Intuitively we expect the implied maturity to increase as the number of options $M$ increases, but the effect is far from linear.  In Figure \ref{M-IM},  we see that the implied maturity is increasing as $M$ increases.  In other words, under the assumption that the ESOs will be exercised gradually, a larger ESO grant has an indirect effect of delaying exercises, and thus leading to higher implied maturity. 
  The increasing trends hold  for different exercise intensities, but the rate of increase diminishes significantly for large $M$. Also, the higher the exercise intensity, the lower the implied maturity.

 \begin{figure}[h]
  \centering
    \includegraphics[width=2.9in]{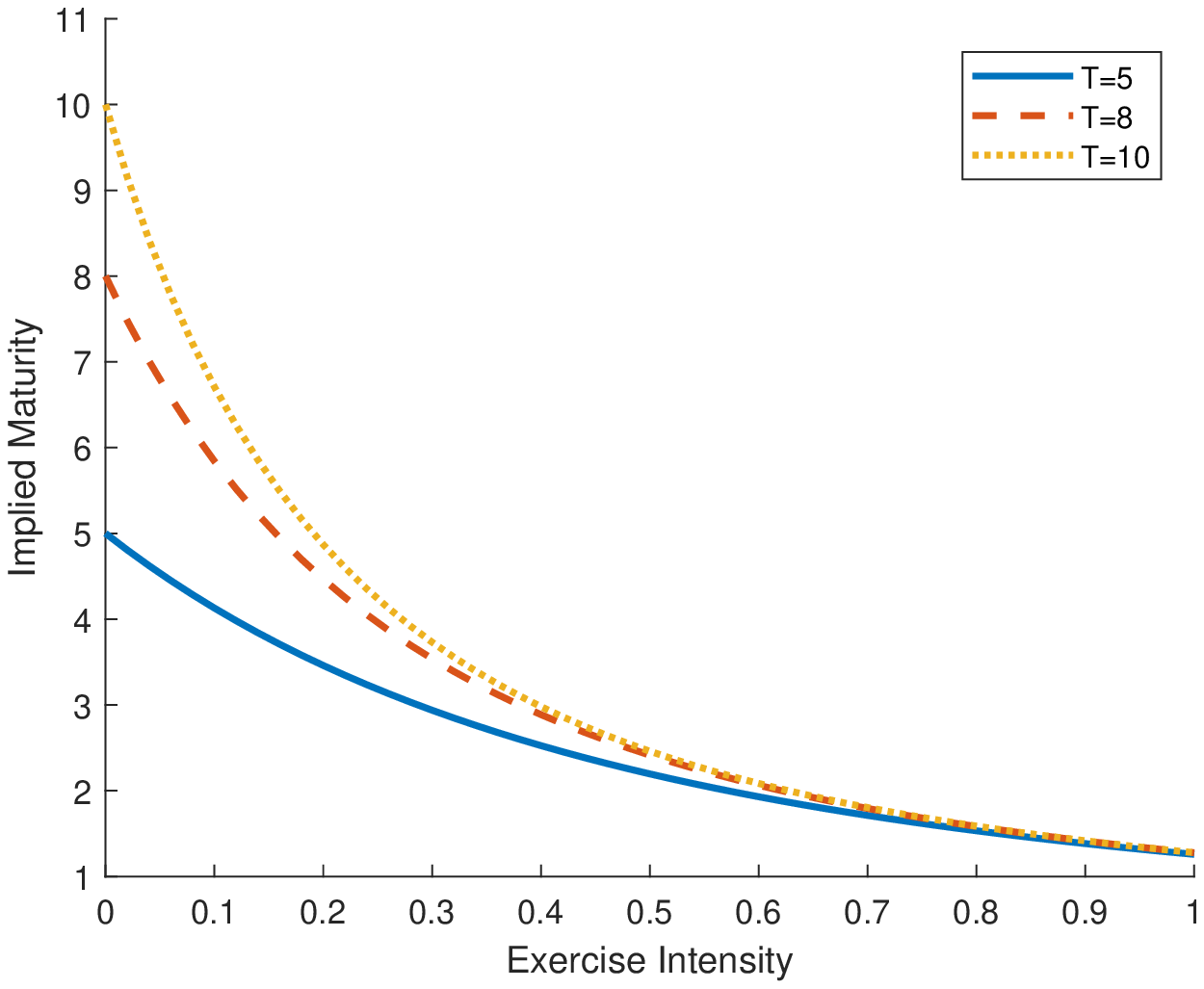}  \includegraphics[width=2.9in]{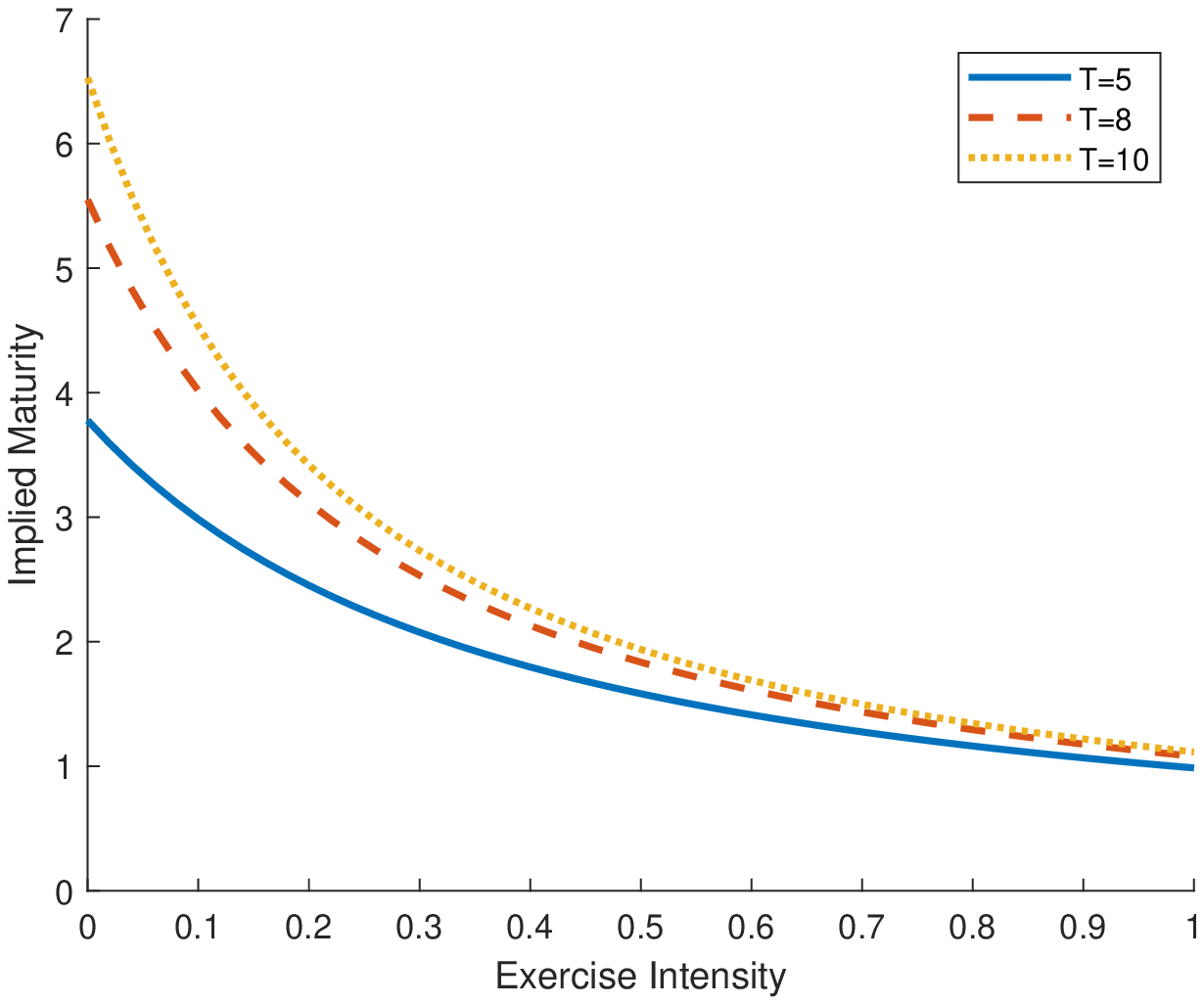}\\
  \caption{Implied maturity as a function of employee exercise intensity  $\lambda$ when the maturity $T= 5, 8$ or $10$, computed using FFT or maturity randomization. Left: FFT.  Right: Maturity randomization.  Parameters: $S_0=K=10$, $r=5\%$, $q=1.5\%$, $\sigma=20\%$, $p_{m,z}=1/m$, $M=5$, $t_v=0$ and $\beta=0.1$. In FFT: $N_x=2^{12}$, $x_{min}=-10$, $x_{max}=10$.}\label{lambda-IM}
\end{figure}

\begin{figure}[h]
  \centering
  \includegraphics[width=2.9in]{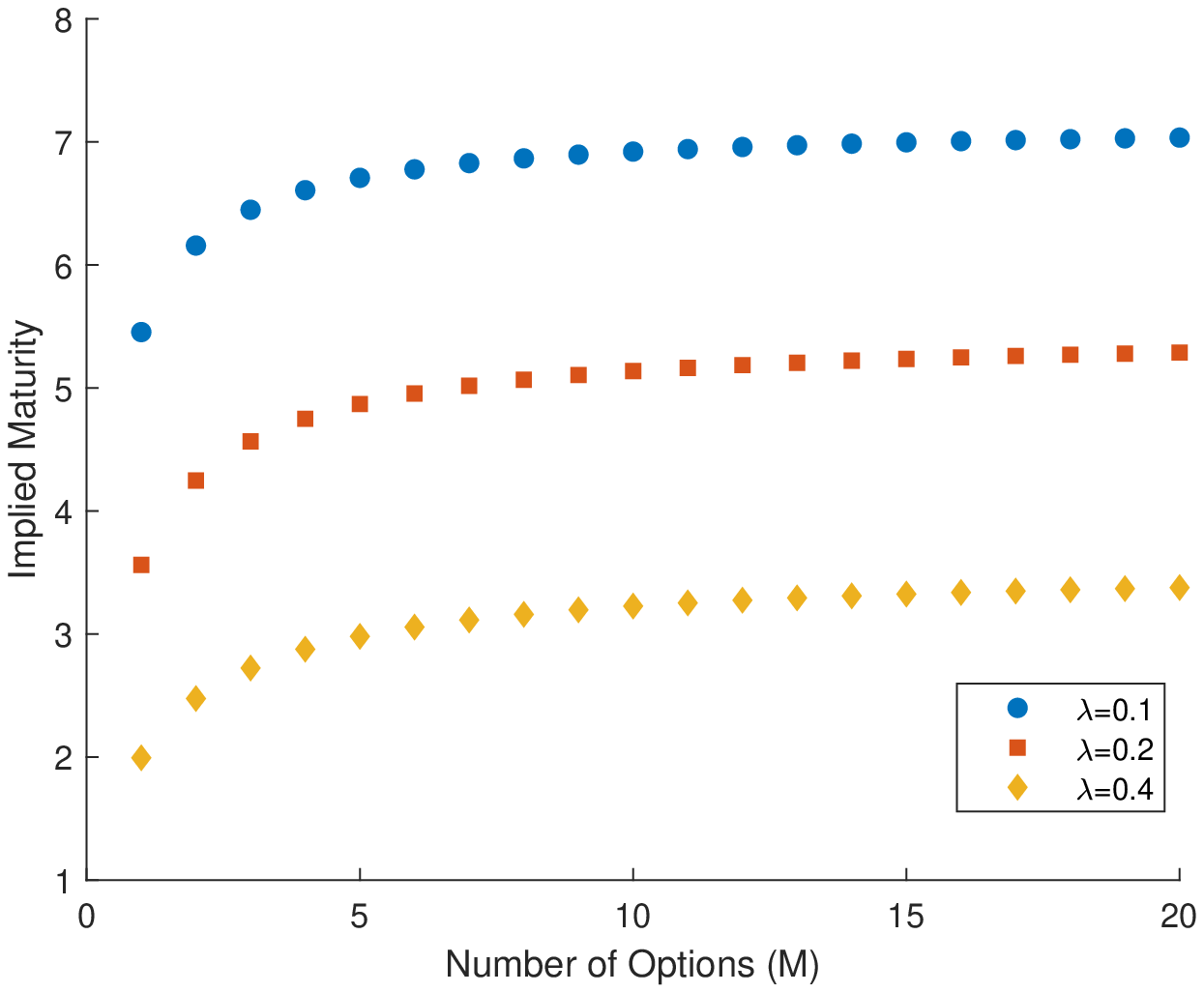}
  \includegraphics[width=2.9in]{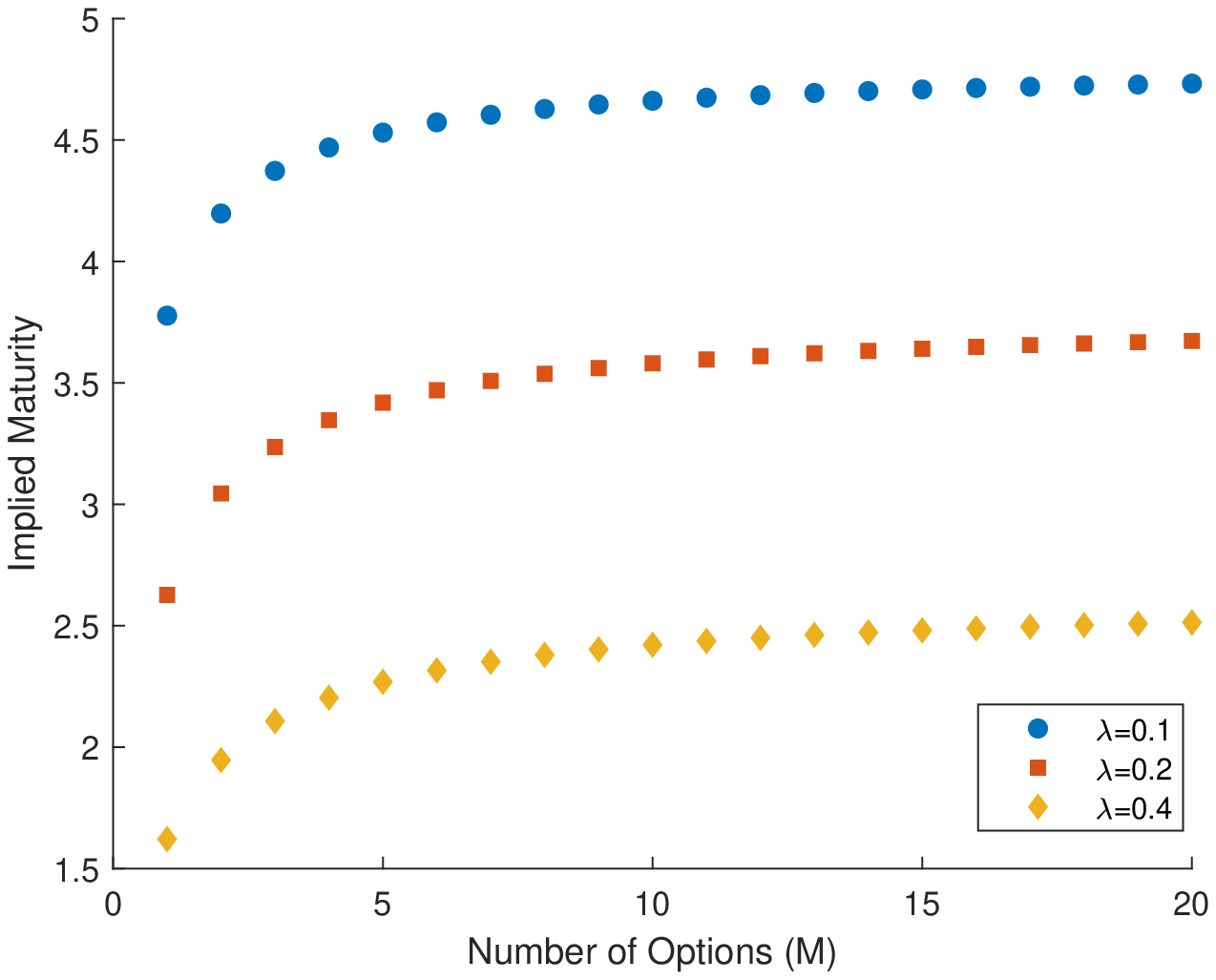}    \caption{Implied maturity as a function of number of options granted  $M$ with different exercise intensities  $\lambda$, computed using FFT or maturity randomization. Left: FFT. Right: Maturity randomization. Common parameters: $S_0=K=10$, $r=5\%$, $q=1.5\%$, $\sigma=20\%$, $p_{m,z}=1/m$, $T=10$, $t_v=0$ and $\beta=0.5$. In FFT: $N_x=2^{12}$, $x_{min}=-10$, $x_{max}=10$.}\label{M-IM}
\end{figure}
 
\clearpage 
 
\section{Conclusion}\label{sect-conclude}
We have studied a new valuation framework that allows the ESO holder to spread out the exercises of different quantities over time, rather than assuming that all options will be exercised at the same time. The holder's multiple random exercises are modeled by an exogenous jump process.  We illustrate the distribution of multiple-date exercises that are consistent with empirical evidence. Additional features included are   job termination risk during and after the vesting period. For cost computation, we apply a fast Fourier transform  method and finite difference method to solve the associated  systems of PDEs. Moreover, we provide an alternative method  based on maturity randomization for approximating the ESO cost. Its analytic formulae for vested and unvested ESO costs allow for instant computation. The proposed numerical method is not only applicable to expensing ESO grants as required by regulators, but also useful for understanding the combined effects of exercise intensity and job termination risk on the ESO cost. For future research, there are a number of directions related to our proposed framework. For many companies, risk estimation for large ESO pool is both practically important and challenging. Another related issue concerns the incentive effect and optimal design of ESOs so that the firm can better align the employee's interest over a longer period of time. 
 
\bibliographystyle{grappa}
\linespread{0}
\begin{small}
\bibliography{mybib_NEW}
\end{small}
\end{document}